\documentclass[11pt,twoside]{article}
\usepackage[utf8]{inputenc}
\usepackage{jeffe}
\usepackage[hmargin=1in,vmargin=1in]{geometry}

\usepackage{pgfplots}
\pgfplotsset{compat=1.14}

\usepackage{cite}
\usepackage{cleveref}
\usepackage{xspace}
\usepackage{enumitem}
\usepackage{amsopn}
\usepackage{verbatim}

\usepackage[mathlines]{lineno}

\DeclareMathOperator{\Argmin}{argmin}

\DeclareMathOperator{\dist}{dist}
\DeclareMathOperator{\probability}{\mathbb{P}}

\def\eps{\varepsilon}
\def\cost{\textsc{Cost}}
\def\spread{\textsc{Sp}}
\def\dartof#1{\vec{#1}}
\def\dartsof#1{\dartof{#1}}

\def\supply{\mu}
\def\map{\tau}
\def\cost{\textsc{cost}}
\def\th{\mathrm{th}}
\def\coor{\mathrm{coor}}
\def\D{\mathbf{d}}
\def\S{s} 
\def\blb{bl} 
\def\bdx{\square} 
\def\cells{\mathbb{C}}
\def\cld{\mathrm{children}}
\def\prnt{\mathrm{parent}}
\def\splitnode{\eta}
\def\splitpoint{r}
\def\splittree{S}
\def\potential{\phi}

\newtheorem{lemma}{Lemma}[section]
\newtheorem{theorem}[lemma]{Theorem}
\newtheorem{corollary}[lemma]{Corollary}

\begin{document}

\begin{titlepage}
  \title{A deterministic near-linear time approximation scheme for geometric transportation}

\author{
Emily Fox%
\thanks{Department of Computer Science,
  The University of Texas at Dallas; \url{emily.fox@utdallas.edu}.
  Supported in part by NSF grant CCF-1942597.}
\and Jiashuai Lu%
\thanks{Department of Computer Science,
  The University of Texas at Dallas; \url{jiashuai.lu@utdallas.edu}.
  This author is currently at Google.
  Supported in part by NSF grant CCF-1942597.}
}


\maketitle

\begin{abstract}
  Given a set of points $P = (P^+ \sqcup P^-) \subset \mathbb{R}^d$ for some constant $d$ and a supply function
  $\mu:P\to \mathbb{R}$ such that $\mu(p) > 0~\forall p \in P^+$, $\mu(p) < 0~\forall p \in
  P^-$, and $\sum_{p\in P}{\mu(p)} = 0$, the geometric transportation problem
  asks one to find a transportation map $\tau: P^+\times P^-\to \mathbb{R}_{\ge 0}$ such that
  $\sum_{q\in P^-}{\tau(p, q)} = \mu(p)~\forall p \in P^+$, $\sum_{p\in P^+}{\tau(p, q)} =
  -\mu(q)~
  \forall q \in P^-$, and the
  weighted sum of Euclidean distances for the pairs $\sum_{(p,q)\in P^+\times P^-}\tau(p, q)\cdot
  ||q-p||_2$ is minimized.
  We present the first deterministic algorithm that computes, in near-linear time, a transportation
  map whose cost is within a $(1 + \varepsilon)$ factor of optimal.
  More precisely, our algorithm runs in $O(n\varepsilon^{-(d+2)}\log^5{n}\log{\log{n}})$ time for
  any constant $\varepsilon > 0$.

  While a randomized $n\varepsilon^{-O(d)}\log^{O(d)}{n}$ time algorithm for this problem was discovered in the last few
  years, all previously known deterministic $(1 + \varepsilon)$-approximation algorithms run
  in~$\Omega(n^{3/2})$ time.
  A similar situation existed for geometric bipartite matching, the special case of geometric
  transportation where all supplies are unit, until a deterministic $n\varepsilon^{-O(d)}\log^{O(d)}{n}$
  time $(1 + \varepsilon)$-approximation algorithm was presented at STOC 2022.
  Surprisingly, our result is not only a generalization of the bipartite matching one to arbitrary
  instances of geometric transportation, but it also reduces the running time for all previously
  known $(1 + \varepsilon)$-approximation algorithms, randomized or deterministic, even for
  geometric bipartite matching.
  In particular, we give the first $(1 + \varepsilon)$-approximate deterministic algorithm for geometric bipartite matching and the first $(1 + \varepsilon)$-approximate deterministic or randomized algorithm for geometric transportation with no dependence on $d$ in the exponent of the running time's polylog.

  As an additional application of our main ideas, we also give the first randomized near-linear\linebreak
  $O(\varepsilon^{-2} m \log^{O(1)} n)$
  time $(1 + \varepsilon)$-approximation algorithm for the uncapacitated minimum cost flow (transshipment) problem
  in undirected graphs with arbitrary \emph{real} edge costs.
\end{abstract}

\setcounter{page}{0}
\thispagestyle{empty}
\end{titlepage}

\pagestyle{myheadings}
\markboth{Emily Fox and Jiashuai Lu}
		{A deterministic near-linear time approximation scheme for geometric transportation}

\section{Introduction}
\label{sec:intro}

Let \(P \subset \R^d\) be a set of \(n\) points in \(d\)-dimensional Euclidean space, and let
\(\supply : P \to \R\) be a function assigning each point a \EMPH{supply} such that \(\sum_{p \in P}
\supply(p) = 0\).
Let \(P^+ = \Set{p \in P \mid \supply(p) > 0}\) and \(P^- = \Set{p \in P \mid \supply(p) < 0}\).
A \EMPH{transportation map} \(\tau : P^+ \times P^- \to \R_{\geq 0}\) is a non-negative assignment
to each ordered pair such that for all \(p \in P^+\) we have \(\sum_{q \in P^-} \map(p, q) =
\supply(p)\) and for all \(q \in P^-\) we have \(\sum_{p \in P^+} \map(p, q) = -\supply(q)\).
A common interpretation of this setting is to imagine each point~\(p \in P^+\) as a pile of earth of
volume \(\supply(p)\) and point \(p \in P^-\) to be a hole of volume \(-\supply(p)\).
A transportation map describes a way to move all of the earth from the piles to the holes.
Accordingly, the \EMPH{cost} of our transportation map is its total \emph{earth-distance} according
to the Euclidean norm, \(\cost(\map) := \sum_{(p, q) \in (P^+ \times P^-)} \map(p, q) \cdot ||q -
p||_2\).%
\footnote{Our results apply to any \(\ell_p\)-norm, but we stick with the \(\ell_2\)-norm to
simplify the presentation.}
Our goal is to find a transportation map of minimum cost \(\cost^*(P, \supply)\), a task we refer
to as the \EMPH{geometric transportation problem}.
Due to the analogy relating the geometric transportation problem to moving piles of earth, the
optimal cost \(\cost^*(P, \supply)\) is often called the \EMPH{earth mover's distance}.
The earth mover's distance is a discrete version of the \emph{\(1\)-Wasserstein distance} between
continuous probability distractions, and the continuous version of this problem has also been
referred to as the \emph{optimal transport} or \emph{Monge-Kantorovich} problem.
Along with being an interesting math problem in its own right, earth mover's distance has
applications to various topics in computer science such as shape matching and graphics
\cite{v-oton-08,gd-fcmua-04,rtg-emdmi-00,bpph-dilmt-11,srgb-emdds-14,cd-fcwb-14,gv-pwps-02,pc-cotad-19,sgpcbndg-cwdeo-15}.

The geometric transportation problem can be viewed as a special case of the minimum cost flow
problem restricted to an uncapacitated complete bipartite graph.
Unfortunately, merely constructing an explicit representation of the appropriate graph
takes~\(\Theta(n^2)\) time.
The transportation map can then be found in strongly polynomial~\(O(n^3 \polylog n)\) time using a
minimum cost flow algorithm of Orlin~\cite{o-fspmc-93}.
If supplies are integral, we can instead use an algorithm of Lee and Sidford~\cite{ls-pfmlp-14} that
will run in~\(O(n^{2.5} \polylog (n, U))\) time where~\(U = \sum_{p \in P} |\supply(p)|\) is the sum
of the absolute values of the points' supplies.
The only faster exact algorithm we are aware of is an implicit implementation of Orlin's algorithm
by~\cite{afpvx-fagtp-17} that runs in~\(O(n^2 \polylog n)\) time and only when given
points in the plane.

Agarwal and Raghvendra~\cite{sa-atpgs-12} described an \(O(n \sqrt{U} \polylog (U, \eps,
n))\)-time\linebreak
\((1+\eps)\)-\emph{approximation} algorithm for the integral supply case.
\cite{anoy-paggp-14} described an~\(O(n^{1 + o(1)})\) time algorithm that computes a
\((1+\eps)\)-factor \emph{estimate} of the earth mover's distance (without associated transportation
map) where the \(o(1)\) hides dependencies on \(\eps\).
Later,~\cite{afpvx-fagtp-17} described a randomized algorithm with an
expected~\(O(\log^2(1 / \eps))\)-approximation ratio that runs in \(O(n^{1 + \eps})\) expected time.
Lahn, Mulchandani, and Raghvendra~\cite{lmr-gtaao-19} described an \(O(n(C \delta)^2 \polylog (U, n))\)-time algorithm that
computes a transportation map of cost at most \(\cost^*(P, \supply) + \delta U\) where \(C = \max_{p
\in P} |\supply(p)|\) is the maximum over the supplies' absolute values.
Finally, Khesin, Nikolov, and Paramonov~\cite{knp-pgtp-21} described a randomized \((1 + \eps)\)-approximation
algorithm with running time \(n \eps^{-O(d)} \log^{O(d)} \spread(P) \log n\) where \(\spread(P)\) is
the \EMPH{spread} or ratio of largest to smallest distance between any pair of points in \(P\).
Via a straightforward reduction, one can use their algorithm to approximately solve the integer
supply case in~\(n \eps^{-O(d)} \log^{O(d)} U \log^2 n\) time~\cite{knp-pgtp-21}.
Fox and Lu~\cite{fl-ntasg-22} subsequently extended their algorithm to run in time \(n \eps^{-O(d)}
\log^{O(d)} n\), a bound which is independent of both the spread and supplies of \(P\).

The above history of the geometric transportation problem neatly mirrors that of the \EMPH{geometric
bipartite matching problem}, the special case of geometric transportation where all supplies are
either \(1\) or \(-1\).
(Geometric bipartite matching also requires the output map to be \(0,1\), but one can guarantee that
is the case with near-linear additional overhead in running time;
see Section~\ref{sec:recover}.)
Indeed, Raghvendra and Agarwal~\cite{sa-ntaag-12,ra-ntaag-20} achieved the same~\(n \eps^{-O(d)}
\log^{O(d)} n\) running time only after a long line of work detailed in their paper.
This running time was recently improved to \(n(\eps^{-O(d^3)}\log \log n + \eps^{-O(d)} \log^4 n \log^5 \log n)\), eliminating the dependence on \(d\) from the polylog's exponent~\cite{arss-iaagb-22}.

One commonality held by many of the above results, including most notably the near-linear time
approximation schemes for geometric bipartite matching and
transportation~\cite{knp-pgtp-21,fl-ntasg-22,ra-ntaag-20,arss-iaagb-22}, is that these results are Monte Carlo
\emph{randomized} algorithms that are guaranteed to work in their reported time bounds but have a
small probability of not finding a good solution.
These four results in particular work by approximating distances between pairs of points using a
sparse graph based on a randomly shifted quadtree.
Agarwal and Raghvendra~\cite{as-aabmm-14} were able to describe a few different
\emph{deterministic} approximation algorithms for geometric bipartite matching with varying
tradeoffs between approximation ratio and running time.
Still, no deterministic \((1 + \eps)\)-approximation algorithm with running time \(o(n^{3/2})\) was
known, even for the geometric bipartite matching problem, for nearly a decade after the initial
publication~\cite{sa-ntaag-12} of Raghvendra and Agarwal's \((1 + \eps)\)-approximation algorithm.

At STOC 2022,~\cite{acrx-dnaag-22} showed that randomness was not necessary for a fast
approximation of geometric bipartite matching by describing a \emph{deterministic} \((1 +
\eps)\)-approximation algorithm that runs in \(n \eps^{-O(d)} \log^{O(d)} n\) time.
Instead of basing distances on a single randomly shifted quadtree, they use the concept of a
\emph{tree cover}, introduced by Awerbuch, Kutten, and Peleg~\cite{akp-obesfdp-94}.
A tree cover can be thought of as \(2^d\) deterministicly shifted quadtrees combined in a way to
guarantee distances are well-approximated.
Through a great deal of effort, they are able to apply the main ideas behind Raghvendra and
Agarwal's~\cite{ra-ntaag-20} algorithm to work in the more complicated setting of a tree cover as
opposed to a single tree.

It is tempting to imagine this same work can be applied to the geometric transportation problem.
Unfortunately, the approach originally taken by Raghvendra and Agarwal~\cite{ra-ntaag-20} and later~\cite{arss-iaagb-22} for
geometric bipartite matching is very different from the one taken by~\cite{knp-pgtp-21}
and~\cite{fl-ntasg-22} for geometric transportation.
The former results iteratively add matching edges along carefully chosen augmenting paths that
increase in length slowly enough that they can all be found in a small amount of time.
The latter results instead build a sparse spanner graph which is entrusted to a minimum cost flow
approximation framework of Sherman~\cite{s-gpumf-17} (see also~\cite{asz-pausp-20}).
Derandomizing the latter results likely requires many ideas other than those used by~\cite{acrx-dnaag-22}, and they presented the design of a fast deterministic approximation
scheme for geometric transportation as an open problem in their paper.

\subsection{Our results}

We describe a \emph{deterministic} \((1 + \eps)\)-approximation algorithm for the geometric
transportation problem that runs in near-linear time.
Specifically, for any \(\eps > 0\), our algorithm is guaranteed to compute a transportation map of
cost at most \((1 + \eps) \cdot \cost^*(P, \supply)\), and it has a worst-case running time of \(O(n
\eps^{-(d+2)} \log^5 n \log \log n\)).

Our approximate transportation map also has the property that each value \(\map(p,q)\) is guaranteed
to be an integer if all supplies~\(\supply(p)\) are integers.
In the special case of all supplies being~\(-1\) and~\(1\), this property implies each
value~\(\map(p,q) \in \Set{0,1}\);
those pairs \((p,q)\) with \(\map(p,q) = 1\) form a matching.
In other words, \(\map\) is a \((1 + \eps)\)-approximate solution to the geometric bipartite
matching problem.\\

We consider our algorithm noteworthy for two main reasons.

\begin{itemize}
  \item
    It derandomizes previous approximation schemes for the geometric transportation problem,
    extending the recent result of~\cite{acrx-dnaag-22} beyond the more specialized geometric
    bipartite matching problem.

  \item
    It actually \emph{improves upon} the running times of all previously known approximation schemes
    for geometric transportation and geometric bipartite matching, both randomized and
    deterministic
    (we are able to remove a \(\log\) factor from the running time for the special case of geometric bipartite matching;
    see Section~\ref{sec:low_spread}).
    In particular, ours is the first \((1 + \eps)\)-approximate deterministic algorithm for geometric bipartite matching and the first \((1 + \eps)\)-approximate determinstic \emph{or randomized} algorithm for geometric transportion where the exponent on the running time's \(\polylog n\) factor is
    bounded by an absolute constant instead of a linear function of the dimension \(d\).
\end{itemize}

\subsubsection*{Application: Approximating uncapacitated minimum cost flow}
Recent work, including~\cite{s-gpumf-17,asz-pausp-20,l-fpaas-20}, has established a connection between
the geometric transportation problem and the minimum cost flow problem in uncapacitated undirected graphs.
The latter is often referred to as the \EMPH{transshipment} problem.
In adjacent papers of the same proceedings, \cite{l-fpaas-20} and \cite{asz-pausp-20} both claim
randomized near-linear \(O(m \polylog n)\) time \((1 + \eps)\)-approximation algorithms for the latter problem.
Unfortunately, both algorithms require the edge costs to have bounded spread or consist of
small integers.
While the running times are polylogarithmic in the spread/sum of edge costs, they can become arbitrarily
high when the values are allowed to be arbitrary real numbers.

That said, \cite{asz-pausp-20} reduces finding a \((1 + \eps)\)-approximation for minimum cost flow to
finding an \(O(\log^{O(1)} n)\)-approximation for a bounded spread instance of geometric transportation in \(O(\log^{O(1)} n)\)-dimensional Euclidean space.
The efficiency of their algorithm crucially relies on both the low spread and the fact that their dependency on the dimension is merely polynomial instead of exponential.
By setting the desired approximation ratio for our deterministic geometric transportation algorithm to be sufficiently large,
we are also able to remove exponential dependencies on dimension while being able to handle point sets of arbitrary spread.
As a consequence, we give the first randomized near-linear time \((1 + \eps)\)-approximation algorithm
for uncapacitated minimum cost flow in undirected graphs with arbitrary real costs.

\subsection{Technical overview}

Similar to how the recent result of~\cite{acrx-dnaag-22} for geometric bipartite
matching uses many ideas originally described by Raghvendra and Agarwal~\cite{ra-ntaag-20}, our
deterministic transportation algorithm is closely tied to the randomized algorithms of~\cite{knp-pgtp-21} and~\cite{fl-ntasg-22}.
We will briefly review their approach and then summarize the new ideas required for its
derandomization.

\subsubsection*{Randomized algorithms}

The randomized transportation algorithms begin by building a randomly shifted quadtree over~\(P\), a
hierarchical collection of \(d\)-dimensional cubic cells where each cell contains at most \(2^d\)
equal sized children and cells containing exactly one point act as leaves.
They then add a large number of Steiner vertices to the collection of input points~\(P\) and use the
tree structure to build a sparse spanner graph over \(P\) and the Steiner vertices.
The number of Steiner vertices added to each cell of the quadtree is \(\Theta((K / \eps)^d)\), where
\(K\) would be the excepted distortion between any pair of vertices if the graph was constructed as
a simple tree containing one Steiner vertex per cell.
In~\cite{knp-pgtp-21}, \(K = O(\log \spread(P))\), while in~\cite{fl-ntasg-22}, \(K = O(\log n)\), and these large numbers of Steiner vertices are
essentially \emph{the} reason why \(O(d)\) appears in the exponents of the runtimes' polylogs.
While there are additional edges, the spanner is essentially a rooted tree where every point and
nearly every Steiner vertex has exactly one \emph{parent} Steiner vertex.
Distances between points of \(P\) are approximately maintained in the spanner graph, implying the
value of an uncapacitated minimum cost flow with sources and sink~\(P\) will serve as a good
estimate for the earth mover's distance (the cost of the optimal transportation map).

Both algorithms use a framework of Sherman~\cite{s-gpumf-17} to approximately compute the minimum
cost flow within the spanner graph.
Sherman's framework requires one to formulate the uncapacitated minimum cost flow problem as finding
a flow vector \(f\) of minimum cost subject to linear constraints \(Af = b\) where \(A\) is the
vertex-edge incidence matrix of the graph and \(b\) is a supply vector not necessarily equal to the
supplies of \(P\).
One repeatedly finds flows \(f\) of approximately optimal cost that approximately satisfy linear
constraints where \(b\) may vary between rounds of the process.
Applied naively, the number of rounds of this process is polynomial in the \emph{condition number}
of \(A\) which can be arbitrarily large.
Fortunately, it is possible to describe a \emph{preconditioner} matrix \(B\) such that \(BA\) has
low condition number.
Repeatedly finding approximate solutions for constraints of the form \(BAf = Bb\) suffices for
finding an approximately optimal solution to the minimum cost flow problem that meets its original
supply constraints exactly.

The preconditioner \(B\) is set up so that \(||Bb||_1\) acts as an estimate on the cost of an
approximately optimal flow \(f\) where \(f\) is found using a very restrictive kind of ``oblivious''
greedy approximation algorithm.
Specifically, the greedy approximation algorithm must move the surplus out of/into any vertex mostly
without regard to the other vertices' surpluses.
The condition number of \(BA\) is proportional to the approximation factor of this greedy
approximation algorithm.
The specific greedy algorithm used in both papers just repeatedly moves the surplus of each vertex
up to its parent.
If a minimum cost flow sends some flow from a vertex \(u\) another vertex \(v\), then the surpluses
pushed up from \(u\) and \(v\) are likely to cancel at a common ancestor not too far away from
either vertex.
Therefore, the cost of pushing these surpluses up is bounded.

When the two algorithms finally have an approximate minimum cost flow that respects the original
supplies of \(P\), they then need to transform it into a proper transportation map without
increasing its cost.
To do so, they shortcut flow to avoid each of the Steiner vertices one by one, using a binary search
tree based data structure to do several shortcuts at once in the case of Fox and
Lu~\cite{fl-ntasg-22}.

\subsubsection*{Derandomization}

We now discuss our techniques for derandomizing the above algorithms.
To make the these discussions easier to follow, we will begin with techniques that lead to a
polylogarithmic dependence on the spread of \(P\) before adding more detail into how we handle the
case of arbitrary spread.

In the previous algorithms~\cite{knp-pgtp-21,fl-ntasg-22}, randomness is directly used only for
picking a randomly shifted quadtree.
The ``obvious'' solution, then, to derandomizing the algorithms is to use a tree cover similar to
how~\cite{acrx-dnaag-22} derandomize the algorithm for geometric bipartite matching.
Indeed, we essentially take this approach in order to construct of our sparse spanner graph.
However, instead of describing it as a collection of \(2^d\) quadtrees with different shifts, it
becomes easier to think of it as a single arbitrary quadtree where each cell is given a single
Steiner vertex, hereinafter referred to as its \emph{net point}, that is directly connected to a
collection of \(O(\eps^{d})\) nearby net points of equal sized cells.
These nearby connections allow for paths to go directly between close-by cells that are not
comparable in the quadtree's hierarchy.
Therefore, we do not need to worry about two close-by points having a distant lowest common ancestor
net point in the tree, and we can guarantee distances are maintained up to a \((1 + \eps)\)-factor
while using a mere \(O((n / \eps^d) \log \spread(P))\) net points total.

The big issue with this approach becomes apparent when attempting to compute minimum cost flows
within Sherman's~\cite{s-gpumf-17} framework.
Our spanner contains edges going between quadtree cells with potentially distant lowest common
ancestors.
Therefore, the greedy approximation algorithm from before that simply pushes surpluses to net
points' parents no longer has an acceptable approximation factor.
In fact, the condition number of \(BA\) may be polynomial in \(\spread(P)\) (and, we emphasize,
\emph{not} polynomial in \(\log \spread(P)\)).
We start to really miss the simplicity afforded us by using a single randomized quadtree without
shortcuts between nearby cells.

The solution to our problem is to simulate random shifting within the greedy algorithm and
preconditioner themselves.
Our greedy approximation algorithm treats the initial surplus of each vertex described by the vector
\(b\) as a separate commodity.
For each vertex~\(u\), for each cell~\(C\) closer to the root of the tree than~\(u\), we explicitly
compute the probability that a random shift of the quadtree would cause the cell~\(C\) to
contain~\(u\).
We then route an equal portion of~\(u\)'s original surplus to~\(C\)'s net point.
Now, suppose a minimum cost flow sends some flow from~\(u\) to another vertex~\(v\).
Using similar algebra to that used in the analyses of the randomly shifted quadtree, we argue that
the cost of surplus from~\(u\) and \(v\) that does not cancel out at any given level is proportional
to the total cost of flow sent from~\(u\) to~\(v\).
Therefore, the approximation factor of the greedy algorithm is proportional to~\(\log \spread(P)\),
the height of the quadtree.
This greedy algorithm is still oblivious enough to imply a preconditioner with condition number
roughly~\(\log \spread(P)\), so we can make Sherman's framework run in near-linear time.

Adding edges between nearby non-related cells also complicates recovering a transportation map from
the approximately optimal flow, because the connected components on each level of the quadtree now
have unbounded size.
We describe a new method of recovering the transportation map that no longer relies on the spanner
having a particular structure.
In fact, our method is general enough that it can be applied to any flow on a spanner over~\(P\),
with or without Steiner vertices.
In short, we process vertices in topological order according to the flow's support graph,
shortcutting flow passing through each vertex we process.
We continue performing the shortcuts through a single vertex in groups using the data structure of
Fox and Lu~\cite{fl-ntasg-22}.

\subsubsection*{Unbounded spread}
The deterministic algorithm for the geometric bipartite matching problem does an \(O(n \log^2 n)\)
time reduction to an instance where the spread is polynomial in \(n\)~\cite{acrx-dnaag-22}.
Applying this reduction with the above observations is already enough to speed up the previous
result;
see Section~\ref{sec:low_spread}.
However, we do require more work to account for unbounded spread cases of geometric transportation.
The main observation used by Fox and Lu~\cite{fl-ntasg-22} to avoid dependencies on spread is that,
with high probability, no input point lies within distance~\(\Delta / \poly n\) of the edge of any
enclosing cell where \(\Delta\) is the side length of that cell (see also~\cite{afpvx-fagtp-17}).
These forbidden regions for cell boundaries are referred to as \emph{moats} around the input points.
Most of their algorithm design and analysis is conditioned on this high probability event.
In particular, the event occuring has two main implications of interest to us:
First, if the set of points within a cell has a bounding box of side length~\(\Delta / \poly n\),
then those points are far enough away from everything else that we can essentially treat them as a
separate instance.
In turn, one can ``compress'' the quadtree so it has height~\(O(n \log n)\) using purely
combinatorial operations.
Second, the expected distortion of distances between any pair of points when using a single Steiner
point per cell is~\(O(\log n)\) instead \(O(\log \spread)\), implying a reasonable approximation
ratio is acheivable by adding extra dependencies on \(\log n\) to the running time.
The gist of their argument is that expected distortion is proportional to the number of quadtree
levels in which a pair of points may be separated with probability strictly between~\(0\) and~\(1\).

To replicate the first implication, we subtly move the axis-aligned hyperplanes bounding cells while
building our quadtree so that no input point lies within the aforementioned \(\Delta / \poly n\) distance from the edge of a cell.
The amount we move the hyperplanes is small enough as to not affect the quality of the spanner, but
it does make it possible to compress the quadtree in a similar manner to Fox and
Lu~\cite{fl-ntasg-22}.
In order to consistently move individual hyperplanes, we first build a collection of binary search
tree based data structures that help us quickly determine whether a particular placement of a
hyperplane would lie too close to one or more points.

We replicate the second implication by modifying how we simulate random shifts during the minimum
cost flow phase of the algorithm.
In short, our goal is to compute probabilities conditioned on shifts not placing cell boundaries too
close to any vertex point.
Suppose we wish to compute how much surplus various net points throughout the spanner should send to
a net point at level~\(\ell\) where cells at level~\(\ell\) have side length~\((1 \pm 1/\poly n)
\Delta\).
We group together maximal collections of net points called \emph{blobs} that cannot be separated by
a shift without one or more of them lying very close to the boundary of a level~\(\ell\) cell.
The surpluses of a blob's constituent points are treated as a single commodity as we figure out how
much surplus to route to each level~\(\ell\) net point.
Now, if a minimum cost flow sends flow from net point~\(u\) to net point~\(v\), there are
only~\(O(\log n)\) levels in which the cost of moving their surpluses is non-negligible; at levels
closer to the root,~\(u\) and~\(v\) appear in the same blob and their surpluses cancel perfectly.
In order to guarantee the uncancelled portions of their surpluses still have cost proportional to
the flow between~\(u\) and~\(v\) across the~\(O(\log n)\) levels that matter, we build a collection
of~\(d\) binary search based data structures that describe the full collection of shifts that
do not separate members of any one blob into distinct level~\(\ell\) cells.
Probability computations are based on the proportional amount of shift allowed according to these
data structures, and we can still use similar algebra to that of the randomized algorithm analyses
to prove approximation quality.

\subsection{Organization}
We describe the construction of the sparse spanner graph in Section~\ref{sec:spanner}.
We describe how to approximate the minimum cost flow within the spanner in Section~\ref{sec:flow}.
We describe how to recover an actual transportation map from the approximate minimum cost flow and
give a theorem stating our main result in Section~\ref{sec:recover}.
We describe some simplifications we can make to our algorithm for the case where
\(\spread(P)\) is sufficiently small in Section~\ref{sec:low_spread}.
These simplifications ultimately result in a slightly lower running time for the special case of geometric bipartite matching.
Finally, we describe our approximation algorithm for uncapacitated minimum cost flow in general undirected
graphs in Section~\ref{sec:general_flow}.

\section{Reduction to minimum cost flow in a sparse spanner graph}
\label{sec:spanner}
In this section, we describe a way to build a sparse spanner graph \(G = (V, E)\) based on a
deterministically constructed quadtree.
As we construct the quadtree, we will subtly shift the hyperplanes bounding its cells so that no
hyperplane goes through a small rectangular \EMPH{moat} around each input point.
At the end of this section, we describe a way to reduce the geometric transportation problem to
finding a minimum cost flow in our graph.

Throughout the rest of this report, we assume without loss of generality that~\(1 / \eps\) is a
power of~\(2\) and that~\(n \geq 1 / \eps\).
We use \(\D\) to denote the set of \(d\) dimensions, and~\(\lg\) to denote the logarithm with
base~\(2\).
As is standard, we will directly prove our algorithm returns a \((1 + O(\eps))\)-approximate
transportation map.
An actual \((1 + \eps)\)-approximation can be obtained in the same asymptotic running time by
dividing~\(\eps\) by a sufficiently large constant.

\subsection{A data structure for avoiding moats}
\label{sec:spanner-data_structures}
Before we begin constructing our quadtree, we need to build a collection of \(d\) data structures
that can be queried to quickly decide if a given axis-aligned hyperplane will intersect any moats of
a given size.
For each dimension~\(i \in \D\), we store a sequence of balanced binary search trees in a persistent
data structure~\cite{dsst-mdsp-89} parameterized by moat size where the nodes of each tree
correspond to maximal collections of points that cannot be separated by the hyperplane orthogonal to
dimension~\(i\).
Each node stores the least \(i\)th coordinate of the points in its collection
along with the greatest \(i\)th coordinate of its points.
Given a value \(x_i\) and a moat size~\(\lambda\), we can easily check in~\(O(\log n)\) time whether the hyperplane orthogonal
to direction~\(i\) with \(i\)th coordinate \(x_i\) intersects a moat, and if so, how far back in
the~\(i\) direction it would need to shift to no longer intersect any moat.
To do so, we do both a predecessor and successor search for~\(x_i\) in the tree for~\(\lambda\).
If~\(x_i\) lies between the two values~\(l_i \leq r_i\) stored for a node, then we need to move
the hyperplane back to~\(l_i - \lambda\).
Otherwise, if its predecessor has highest coordinate~\(r_i\) and~\(r_i + \lambda > x_i\), we
let~\(l_i\) be the least coordinate of the predecessor node and move the hyperplane to~\(l_i -
\lambda\).
Finally, if the successor has least coordinate \(l_i\) with \(l_i - \lambda < x_i\), we again move
the hyperplane to~\(l_i - \lambda\).

To build the data structure for dimension~\(i\), we begin by building the tree for moat
size~\(\lambda = 0\):
it is simply a balanced binary search tree over the members of~\(P\) where their \(i\)th coordinates
act as the keys.
We now act as if~\(\lambda\) is continuously increasing, processing the next \EMPH{event
moment}~\(\lambda\) where the moats around two consecutive nodes' points meet.
These event moments can be computed and ordered in advance in~\(O(n \log n)\) time by looking at the
difference in \(i\)-coordinate between consecutive points and sorting.
At each event moment, we remove the two nodes whose moats are meeting and replace them with a single
node.
It takes~\(O(n \log n)\) time total to processes all the events.
Again, we store the sequence of trees in a persistent data structure so we can easily access the
current tree for any given value~\(\lambda\) in~\(O(\log n)\) time.
We require~\(O(\log n)\) additional space per tree in the sequence, for~\(O(n \log n)\) space total.

\begin{lemma}
  In~\(O(n \log n)\) time, we can create a collection of~\(d\) data structures using~\(O(n \log n)\)
  space each so that, for any given dimension~\(i \in \D\), coordinate~\(x_i\), and moat
  size~\(\lambda \geq 0\), we can lookup in~\(O(\log n)\) time whether the hyperplane orthogonal to
  dimension~\(i\) at~\(x_i\) intersects any point's moat, and if so, how far back in the~\(i\)th
  dimension it needs to be shifted to avoid hitting any moats.
\end{lemma}

\subsection{Warped quadtree}
With the preprocessing out of the way, we can turn to constructing the spanner itself.
We begin by building what we call a \EMPH{warped quadtree} \(T\) on \(P\).
Let \(\square_P\) be the minimum bounding hypercube containing \(P\).
Let~\(C^*\) be the hypercube centered at the center of~\(\square_P\) but with twice its side length.
Warped quadtree~\(T\) has root cell~\(C^*\).
The other cells of~\(T\) are not necessarily hypercubes, but we do guarantee each cell is an
axis-parallel box.
We generally use~\(\Delta_{C, i}\) to denote the length of a cell~\(C\) in the \(i\)th
dimension and~\(\ell_{C}\) to denote its \EMPH{level} or number of edges on the unique path in \(T\)
from~\(C^*\) to~\(C\).

We construct~\(T\) iteratively as follows.
We first add~\(C^*\) to a queue of unprocessed cells.
While there exists a cell~\(C\) we have not yet processed, we remove~\(C\) from the queue and
perform the following steps.
If there are~\(\lg (n^2 / \eps)\) ancestor cells of~\(C\) including~\(C\) itself that all contain a
single point from~\(P\), then~\(C\) is a leaf.
We are done processing~\(C\).

Otherwise, let~\(P' \subseteq P = C \cap P\), and let~\(\square_{P'}\) denote the minimum bounding
\emph{hypercube} containing~\(P'\).
Let~\(\Delta_{C} = \min_i \Delta_{C,i}\) and~\(\Delta_{P'}\) denote the side length of
\(\square_{P'}\).
If~\(|P'| \geq 2\) and~\(\Delta_{P'} < \Delta_{C} / n^4\), we \emph{contract}~\(P'\) to a single
point~\(p \in P'\) as described below before moving on to the next steps in processing~\(C\).

Now, we partition~\(C\) into~\(2^d\) \emph{approximately} equal sized axis-aligned boxes by
splitting~\(C\) along the following~\(d\) hyperplanes.
For each dimension~\(i \in \D\), the~\(i\)th hyperplane lies orthogonal to dimension~\(i\).
Let \(x_i\) be the average of the \(i\)th coordinates for the two bounding faces of~\(C\) lying orthogonal
to dimension~\(i\).
We query the moat-avoiding data structure for dimension~\(i\) and place hyperplane~\(i\) at 
the coordinate the data structure says we should use
instead of~\(x_i\) so it does not intersect any moat of size \(\lambda = \Delta_{C, i} / (2n^2)\).

As stated, the~\(d\) hyperplanes partition~\(C\) into~\(2^d\) boxes.  For each such box~\(C'\) such
that~\(C' \cap P \neq \emptyset\), we add~\(C'\) as a child of~\(C\) and add~\(C'\) to the queue of
unprocessed cells.  We are done processing~\(C\).

We now specify how to \EMPH{contract} a subset of points~\(P' \subseteq P\) as mentioned above.
Let~\(p \in P'\) be an arbitrary member of~\(P'\).
We create a new instance of the geometric transportation problem~\((P', \supply')\) such that
\(\supply'(q) = \supply(q)\) for all \(q \in P' \setminus \{p\}\) and \(\supply'(p) = -\sum_{q \in
(P' \setminus \{p\})} \supply(q)\).
Finally, we remove all points in~\(P' \setminus \{p\}\) from~\(P\) and change~\(\supply(p)\) to be
\(\sum_{q \in P'} \supply(q)\), the total supply of all points in~\(P'\), including the original
supply of~\(p\).
We do not modify the currently constructed tree~\(T\) or the data structures for avoiding moats when
we modify~\(P\).

Later, we describe how to build a sparse spanner graph~\(G\) over the contracted set of points~\(P\)
(note that we may perform further contractions to~\(P\) before we actually construct~\(G\)).
We then compute a \((1 + O(\eps))\)-approximately minimum cost flow \(f\) on~\(G\).
The last component of our contraction procedure is to recursively compute one or more spanners
for~\((P',\supply')\) and \((1 + O(\eps))\)-approximate flows on those spanners.
In Section~\ref{sec:recover}, we recover an approximately optimal transportation map from the union
of these flows.
The following lemma will be of use when we analyze the total cost of these flows and our final
transportation map.
\begin{lemma}\label{lem:contraction}
  Let~\((P^/, \supply^/)\) denote the problem instance~\((P, \supply)\) after contracting~\(P'\).
  We have\linebreak
  \(\cost^*(P^/, \supply^/) + \cost^*(P', \supply') \leq (1 + O(1 / n^2)) \cost^*(P,
  \supply)\).
\end{lemma}

\begin{proof}
  Let~\(\map\) be an arbitrary transportation map for~\((P, \supply)\), and let~\(p \in P'\) be the
  point replacing~\(P'\) during its contraction.
  We will construct two transportation maps~\(\map^/\) and~\(\map'\) for~\((P^/, \supply^/)\)
  and~\((P', \supply')\), respectively, such that~\(\cost(\map^/) + \cost(\map') \leq (1 + O(1 /
  n^2)) \cost(\map)\).
  For all \((a, b) \in ((P^/)^+ \times (P^/)^-)\), \(a,b \neq p\), we set~\(\map^/(a, b) := \map(a,
  b)\).
  Similarly, for all \((a, b) \in ((P')^+ \times (P')^-)\), \(a, b \neq p\), we set \(\map'(a,
  b) := \map(a, b)\).
  For all \(q \in ((P^/)^+ \setminus \{p\})\), we set \(\map^/(q, p) := \sum_{r \in (P')^-} \map(q,
  r)\), and for all \(q \in ((P^/)^- \setminus \{p\})\), we set \(\map^/(p, q) := \sum_{r \in
  (P')^+} \map(r, q)\).
  Note that we might now have non-zero pair assignments with~\(p\) in both the first and second
  position, but we can shortcut ``flow'' going through~\(p\) to make~\(\map^/\) a valid
  transportation map while only reducing its cost.
  Similarly, for all \(q \in ((P')^+ \setminus \{p\})\), we set \(\map'(q, p) := \sum_{r \in
  (P^/)^-} \map(q, r)\), and for all \(q \in ((P')^- \setminus \{p\})\), we set \(\map'(p, q) :=
  \sum_{r \in (P^/)^+} \map(r, q)\), again shortcutting as necessary to make~\(\map'\) a valid
  transportation map.

  We now verify our claim on the total cost of \(\map^/\) and \(\map'\).
  Consider any \(q \in (P^/ \setminus \{p\})\) and \(r \in (P' \setminus \{p\})\).
  Our algorithm contracts~\(P'\) while processing a cell~\(C\).
  By construction of~\(C\), every point in~\(P'\) is distance at least~\(\Delta_{C} /
  n^2\) from the boundary of~\(C\), and our choice to contract implies the diameter of~\(P'\)
  is less than~\(\sqrt{d} \Delta_{C} / n^4\).
  Therefore, \(||p - q||_2 + ||r - p||_2 \leq (1 + O(1 / n^2)) ||r - q||_2\).
  To keep the algebra concise, we define~\(\map^/(q,r)\) or~\(\map'(q,r)\) to be~\(0\)
  whenever~\((q,r)\) is not really in the domain of~\(\map^/\) or~\(\map'\).
  We see

  \begin{align*}
    \cost(&\map^/) + \cost(\map')&& \\
    &= \sum_{(q, r) \in (P^/ \times P^/)} \map^/(q, r) \cdot ||r -
    q||_2 + \sum_{(q, r) \in (P' \times P')} \map'(q, r) \cdot ||r - q||_2 \\
    &= \sum_{(q, r) \in ((P^/ \times P^/) \cup (P' \times P')) \mid q,r \neq p} \map(q, r) \cdot ||r
    - q||_2
    + \sum_{q \in (P^/ \setminus \{p\})} \Paren{\map^/(p, q) + \map^/(q, p)} ||p - q||_2 \\
    &\qquad\qquad + \sum_{q \in (P' \setminus \{p\})} \Paren{\map'(p, q) + \map'(q, p)} ||p - q||_2
    \\
    &\leq \sum_{(q, r) \in ((P^/ \times P^/) \cup (P' \times P')) \mid q,r \neq p} \map(q, r) \cdot
    ||r - q||_2
    + \sum_{q \in (P^/ \setminus \{p\})} \sum_{r \in P'} \Paren{\map(r, q) + \map(q,
    r)} \cdot ||p - q||_2 \\
    &\qquad\qquad + \sum_{q \in (P' \setminus \{p\})} \sum_{r \in P^/} \Paren{\map(r, q) + \map(q,
    r)} \cdot ||p - q||_2 \\
    &= \sum_{(q, r) \in ((P^/ \times P^/) \cup (P' \times P'))} \map(q, r) \cdot ||r - q||_2 \\
    &\qquad\qquad + \sum_{q \in (P^/ \setminus \{p\})} \sum_{r \in (P' \setminus \{p\})}
    \Paren{\map(r, q) + \map(q, r)} \cdot (||p - q||_2 + ||r - p||_2) \\
    &\leq \sum_{(q, r) \in ((P^/ \times P^/) \cup (P' \times P'))} \map(q, r) \cdot ||r - q||_2 \\
    &\qquad\qquad + \sum_{q \in (P^/ \setminus \{p\})} \sum_{r \in (P' \setminus \{p\})}
    \Paren{\map(r, q) + \map(q, r)} \cdot (1 + O(1 / n^2)) ||r - q||_2 \\
    &\leq (1 + O(1 / n^2)) \cost(\map).
  \end{align*}
\end{proof}

\subsection{Properties of warped quadtrees}
We now prove some basic properties of the warped quadtree~\(T\).

\begin{lemma}\label{lem:num_cells}
  Suppose~\(P\) has~\(n'\) points remaining after all contractions used in the construction
  of~\(T\).
  Warped quadtree~\(T\) contains at most~\(O(n' \log n)\) cells.
\end{lemma}
\begin{proof}
  Consider a path of cells~\(\langle C_1, C_2, \dots, C_k \rangle\) where each cell~\(C_j\) in the
  path contains the same point subset~\(P'\).
  If~\(|P'| = 1\), then~\(k \leq \lg (n^2 / \eps) = O(\log n)\).
  Now suppose otherwise.
  For each \(j \in \Set{2, \dots, k}\), we have \(\Delta_{C_j} \leq ((1/2 + 1/(2n^2))
  \Delta_{C_{j-1}}\).
  Therefore, \(\Delta_{P'} \leq \Delta_{C_k} \leq (1/2 + 1/(2n^2))^{k - 1} \Delta_{C_1}\).
  On the other hand, \(\Delta_{P'} \geq \Delta_{C_1} / n^4\), because we did not contract~\(P'\) to
  one point.
  We again conclude that~\(k = O(\log n)\).
  We complete the proof by recalling there are at most~\(n' - 1\) cells~\(C\) where each child
  of~\(C\) contains strictly fewer points than~\(C\).
\end{proof}

\begin{lemma}\label{lem:tree_time}
  Let~\(m\) be the total number of cells in all warped quadtrees, including those constructed
  recursively during contractions.
  We can construct all the warped quadtrees in~\(O((m + n) \log n)\) time total.
\end{lemma}
\begin{proof}
  We use a similar strategy to that used in prior work~\cite{ck-fasgg-93,fl-ntasg-22}, complicated
  only by the existence of our data structures for avoiding moats.
  Given the original input point set~\(P\), we begin by creating \(d\) doubly-linked lists of the
  points~\(P\), each sorted by a different coordinate.
  As we process any cell~\(C\) with point subset~\(P'\), we will provide access to a sorted sublists
  containing just the points~\(P'\).
  We will also provide the total supply of points in~\(P'\).

  Suppose we wish to process a cell~\(C\).
  If~\(|P'| = 1\), we check if the \(\lg (n^2 / \eps) - 1\) first strict ancestors of~\(C\) have
  other children in~\(O(\log n)\) time, and if not, declare~\(C\) to be a leaf and stop processing
  it.
  If not, we use the sorted lists to determine whether~\(\Delta_{P'} < \Delta_{C} / n^4\) in
  constant time.
  If so, let~\(p\) be the point chosen to represent~\(P'\) after contraction.
  We can use the total supply of~\(P'\) to compute the new supplies of~\(p\) in both the current and
  recursive instances of the geometric transportation problem in constant time.
  Further, we directly hand off the sorted lists for~\(P'\) to the recursive instance so that the
  root cell of the recursive instance can be computed in constant time as well.

  After the possible contraction, we still need to find the child cells of~\(C\).
  For each dimension \(i \in \D\), we use its moat avoiding data structure to compute the coordinate
  for the axis-aligned hyperplane orthogonal to~\(i\) in~\(O(\log n)\) time.
  We then search that dimension's linked list of points from both ends in simultaneously to find
  where that hyperplane splits in the points in time proportional to the number of points in the
  less populated side of the split.
  We also perform individual deletions and insertions of the other dimensions' linked lists to
  create the sorted lists of every dimension for the less populated side in time proportional to its
  number of points.
  The remains of the original linked lists are the sorted point sets for the more populated side.
  When we have finished splitting along each dimension, we add all cells with at least one point to
  be children of~\(C\) and then add them to the queue of cells to process.

  Outside of splitting the point sets, we spend at most~\(O(\log n)\) time per each of the~\(m\)
  cells.
  The total time spent doing splits throughout all cells is proportional to the total number of
  points going to the less populated sides of splits.
  Every time a point goes to a less populated cell, it shares that cell with at most half as many
  points as it did before, meaning we move each point to a less populated cell at most~\(O(\log n)\)
  times.
  The total time computing splits is therefore~\(O(n \log n)\), and we spend~\(O((m + n) \log n)\)
  time total computing all warped quadtrees.
\end{proof}

\begin{lemma}\label{lem:level_length}
  Any cell at level~\(\ell\) has sides of length at least \(\Delta_{C^*} / (2^{\ell} e)\) and at
  most \((e \Delta_{C^*}) / 2^{\ell}\) where~\(e \approx 2.72\) denotes Euler's number.
\end{lemma}
\begin{proof}
  Assuming~\(n\) is sufficiently large, Lemma~\ref{lem:num_cells} implies warped quadtree~\(T\) has
  height strictly less than~\(n^2\).
  Let~\(i \in \D\) be any dimension.
  For any cell~\(C\) and child cell~\(C'\), we have \((1/2 - 1/(2n^2))\Delta_{C,i} \leq
  \Delta_{C',i} \leq (1/2 + 1/(2n^2)) \Delta_{C,i}\).
  Therefore, the minimum side length at level~\(\ell\) is at least
  \[\Paren{\frac{1}{2} - \frac{1}{2n^2}}^{\ell} \Delta_{C^*} > \frac{\Delta_{C^*}}{2^{\ell}}
  \Paren{1 - \frac{1}{n^2}}^{n^2} > \frac{\Delta_{C^*}}{2^{\ell} e}, \]
  and the maximum side length is at most
  \[\Paren{\frac{1}{2} + \frac{1}{2n^2}}^{\ell} \Delta_{C^*} < \frac{\Delta_{C^*}}{2^{\ell}}
  \Paren{1 + \frac{1}{n^2}}^{n^2} < \frac{e \Delta_{C^*}}{2^{\ell}}. \]
\end{proof}

\subsection{Constructing the spanner}\label{subsec:build-spanner}
We now describe how to build our sparse spanner graph~\(G = (V, E)\) using the warped quadtree
described above.
For each cell~\(C\) in~\(T\), we add a single \EMPH{net point}~\(N_C\) at the center of~\(C\).
The vertices~\(V\) of~\(G\) constitute the set~\(P\) unioned with the set of net points.

We add an edge from each point~\(p \in P\) to the net point~\(N_C\) where~\(C\) is the leaf cell
containing~\(p\).
For each pair of net points \(N_C\) and \(N_{C'}\) such that \(C'\) is the parent cell of \(C\), we
add an edge between \(N_C\) and \(N_{C'}\).
By construction of~\(T\), the cells at any single level~\(\ell\) form a subset of a
\(d\)-dimensional grid, albeit with somewhat uneven spacing between consecutive parallel boundary
hyperplanes.
For each cell~\(C\) belonging to some level~\(\ell\), we add edges from~\(N_C\) to the at most~\((1
/ \eps + 1)^d - 1 = O(1 / \eps^d)\) other net points~\(N_{C'}\) where~\(C'\) is another
level~\(\ell\) cell at most~\(1 / (2\eps)\) grid cells away in each of the \(d\) dimensions.
All edges are weighted according to the Euclidean distance between their endpoints.
We let~\(\dist_G(p, q)\) denote the shortest path distance between vertices~\(p\) and~\(q\)
in~\(G\).

\begin{lemma}
  Let~\(n'\) be the number of points in~\(P\) after contractions.
  The sparse spanner graph~\(G = (V, E)\) has \(O(n' \log n)\) vertices and \(O((n' / \eps^d) \log
  n)\) edges.
  Given the warped quadtree~\(T\), it can be built in~\(O((n' / \eps^d) \log n)\) time.
\end{lemma}
\begin{proof}
  The number of vertices follows immediately from Lemma~\ref{lem:num_cells}.
  We add~\(O(1 / \eps^d)\) edges per net point, establishing the claimed total number of edges.
  We do the following to construct~\(G\).
  Recall, the cells at any particular level form a subset of a grid.
  For each level~\(\ell\), we sort its cells by their location in their grid so we may efficiently
  find all adjacent pairs of cell.
  We then search the neighborhood around each cell to figure out which edges to add to that cell's
  net point.
  Each level contains at most~\(n\) cell, so the total time spent sorting at all levels is~\(O(n'
  \log n)\).
  The time spent searching neighbors is proportional to the number of edges in~\(G\).
\end{proof}

\begin{lemma}\label{lem:approximate_distances}
  The distance between any pair of points \(p,q\in P\) in \(G\) is at most \((1+O(\eps)) \cdot ||p -
  q||_2\).
\end{lemma}
\begin{proof}
  Let~\(C\) be the deepest/smallest cell containing both~\(p\) and~\(q\).
  There exists at least one axis-aligned hyperplane splitting~\(C\) into children cells that
  specifically separates~\(p\) and~\(q\).
  Let this hyperplane be orthogonal to dimension~\(i\).
  By construction, both~\(p\) and~\(q\) are distance~\(\Delta_{C, i} / n^2\) or more from this
  hyperplane.
  Applying Lemma~\ref{lem:level_length}, we conclude~\(||q - p||_2 \geq \Delta_{C^*} / (2^{\ell_C -
  1} e n^2)\).

  Let~\(C_p\) and~\(C_q\) be the leaf cells containing~\(p\) and~\(q\), respectively.
  Both leaf cells lie at level~\(\ell_C + \lg (n^2 / \eps)\) or higher.
  Therefore, they have side lengths at most \((\eps e \Delta_{C^*}) / (2^{\ell_C} n^2) \leq ((\eps
  e^2) / 2) \cdot ||q - p||_2\).

  Let~\(\ell'\) denote the greatest level where either \(C_p\) and \(C_q\) have the same ancestor at
  level~\(\ell'\) or the net points of their level~\(\ell'\) ancestors are adjacent in~\(G\).
  Let~\(C'_p\) and~\(C'_q\) denote the ancestors of~\(C_p\) and~\(C_q\), respectively, at
  level~\(\ell'\).
  Suppose \(C'_p \neq C_p\) and \(C'_q \neq C_q\).
  In that case, their children cells containing~\(p\) and~\(q\), respectively, do not contain
  adjacent net points.
  There must be at least \(1 / (2\eps)\) level \(\ell' + 1\) cells separating those two children in
  some dimension, implying \(||q - p||_2 \geq \Delta_{C^*} / (2^{\ell' + 2} e \eps)\).
  Meanwhile, the sides of~\(C'_p\) and~\(C'_q\) have length at most \((e \Delta_{C^*}) / 2^{\ell'}
  \leq (4 e^2 \eps) \cdot ||q - p||_2\).
  We conclude, whether or not one of \(C'_p\) or \(C'_q\) is a leaf, they have sides of length at
  most \((4 e^2 \eps) \cdot ||q - p||_2\).

  Let \(N'_p = N_{C'_p}\) and \(N'_q = N_{C'_q}\) denote the net points of \(C'_p\) and \(C'_q\),
  respectively.
  Triangle inequality implies \(\dist_G(N'_p, N'_q) = ||N'_q - N'_p||_2 \leq ||p - N'_p||_2 + ||q -
  p||_2 + ||N'_q - q||_2 \leq (1 + 4 \sqrt{d} e^2 \eps) \cdot ||q - p||_2\).
  Some (admittedly loose) algebra based on the diameters of the descendent cells of \(C'_p\) and
  \(C'_q\) implies both \(\dist_G(p, N'_p)\) and \(\dist_G(q, N'_q)\) to be at most \((\sqrt{d} / 2)
  \sum_{k = 0}^{n^2} (1/2 + 1/(2n^2))^k ((4 e^2 \eps) \cdot ||q - p||_2) \leq (4 \sqrt{d} e^3 \eps)
  \cdot ||q - p||_2\).
  Finally, we see \(\dist_G(p, q) \leq \dist_G(p, N'_p) + \dist_G(N'_p, N'_q) + \dist_G(N'_q, q)
  \leq (1 + 4 \sqrt{d} (e^2 + e^3) \eps) \cdot ||q - p||_2 = (1 + O(\eps)) \cdot ||q - p||_2\).
\end{proof}

\subsection{Reduction to minimum cost flow}
We are now ready to reduce the problem of computing an approximately optimal transportation map for
contracted instance \((P, \supply)\) to one of computing an approximately optimal minimum cost flow
in our sparse spanner graph~\(G = (V, E)\).
Our formulation of the uncapacitated minimum cost flow problem follows prior
work~\cite{knp-pgtp-21,fl-ntasg-22}.

Let~\(\dartsof{E}\) be the set of edges~\(E\) oriented arbitrarily.
A vector~\(f \in \R^{\dartsof{E}}\) indexed by the oriented edges~\(\dartsof{E}\) is called a
\EMPH{flow vector} or often simply a \EMPH{flow}.
Let \(A\) denote the \(|V| \times |\dartsof{E}|\) \EMPH{vertex-edge incidence matrix} where for each
vertex-edge pair \((u, (v, w)) \in V \times \dartsof{E}\), we have \(A_{u, (v,w)} = 1\) if \(u =
v\), \(A_{u, (v,w)} = -1\) if \(u = w\), and \(A_{u, (v,w)} = 0\) otherwise.
The \EMPH{divergence} of a flow \(f\) at a vertex \(v\) is defined as \((Af)_v = \sum_{(v, w)}
f_{(v,w)} - \sum_{(u,v)} f_{(u,v)}\).
For each edge \((u,v) \in \dartsof{E}\), we abuse notation by letting \(f_{(v,u)} := -f_{(u,v)}\).

Let \(|| \cdot ||_{\dartsof{E}}\) denote the norm on \(\R^{\dartsof{E}}\) where
\(||f||_{\dartsof{E}} = \sum_{(u,v) \in \dartsof{E}} |f_{(u,v)}| \cdot ||v - u||_2\).
We define an instance of the \EMPH{uncapacited minimum cost flow problem} in our spanner graph \(G\)
as a pair \((G, b)\) where \(b \in \R^V\) is a given set of vertex divergences.
A feasible solution to the problem is a flow vector \(f\) such that \(Af = b\).
Let \(\cost^*(G, b)\) to be the minimum value \(||f||_{\dartsof{E}}\) among feasible flow
vectors~\(f\).
The goal is to find a flow vector achieving this minimum.

For our reduction from the geometric transportation problem, we define \(b^* \in \R^V\) to be the
set of divergences such that \(b^*_p = \supply(p)\) for all \(p \in P\) and \(b^*_v = 0\) for all
net points~\(v\).
By the construction of~\(G\) and Lemma~\ref{lem:approximate_distances}, we have \(\cost^*(P,
\supply) \leq \cost^*(G, b^*) \leq (1 + O(\eps)) \cost^*(P, \supply)\).
Our goal in Section~\ref{sec:flow} is to compute a feasible flow~\(f\) for \((G, b^*)\) of cost
\(||f||_{\dartsof{E}} \leq (1 + O(\eps)) \cdot \cost^*(G, b^*) = (1 + O(\eps)) \cost^*(P,
\supply)\).

We perform at most~\(n - 1\) contraction operations across the whole of our algorithm.
Lemma~\(\ref{lem:contraction}\) implies the total cost of optimal transportation maps across all
contracted instances form a \((1 + O(1 / n))\)-approximation of the optimal cost for the original
input point set.
We build spanners for all of these instances, and use them to find flows of cost at most \((1 +
O(\eps))\) times the optimal cost of a transportation map for each of these instances.
The total cost of these flows is a \((1 + O(\eps))\)-approximation of the optimal cost for the
original input transportation instance.
If~\(n'\) is the size of any of these minimal point sets, its approximate flow will be computed
in~\(O(n' \eps^{-(d+2)} \log^5 n \log \log n)\) time, implying computing and combining all the
individual flows will take \(O(n \eps^{-(d+2)} \log^5 n \log \log n)\) time.
Finally, in Section~\ref{sec:recover}, we turn this combined flow into a transportation map of no
greater cost in~\(O(n \log^2 n)\) time.


\section{Preconditioning for minimum cost flow}
\label{sec:flow}
Let~\(G = (V, E)\) be the spanner defined in Section~\ref{sec:spanner} for contracted geometric
transportation instance~\((P, \supply)\), and let~\(b^*\) be the set of divergences defined for
the corresponding instance of minimum cost flow.
In this section, we describe a way to find a \((1 + O(\eps))\)-approximate solution for the minimum
cost flow instance \((G, b^*)\) using Sherman's generalized preconditioning
framework~\cite{s-gpumf-17}.

Let \(C_p\) be denote the leaf cell containing a point \(p\in P\).
By the definition of \(G\), point \(p\) is incident to exact one edge connecting \(p\)
to~\(N_{C_p}\).
For simplicity, let \(f'\) be a flow such that, for all points \(p \in P\), \(f'_{(p, N_{C_p})} =
b^*_p\).
From now on, we assume \(G = (V, E)\) does not have any point in \(P\), and let \(b\in \R^V\) such
that \(b_{N_{C_p}}=b^*_p, \forall p\in P\) and \(b_v = 0\) otherwise.
We focus on finding an \((1 + O(\eps))\)-approximation \(f\) of the minimum cost flow instance on
\((G, b)\).
The flow \(f + f'\) is then a \((1 + O(\eps))\)-approximate solution for the minimum cost flow
instance \((G, b^*)\).

Consider any instance of the minimum cost flow problem in \(G\) with an arbitrary divergence vector
\(\tilde{b} \in \R^{V}\), and let \(f^*_{\Tilde{b}}:=\Argmin_{f\in\R^{\dartsof{E}},
Af=\Tilde{b}}{||f||_{\dartsof{E}}}\).
A flow vector \(f\in\R^{\dartsof{E}}\) is an \EMPH{\((\alpha, \beta)\)-solution} to the problem if
\begin{align}
    \nonumber ||f||_{\dartsof{E}}&\le \alpha||f^*_{\Tilde{b}}||_{\dartsof{E}}\\
    \nonumber ||Af-\tilde{b}||_1&\le \beta||A||\,||f^*_{\Tilde{b}}||_{\dartsof{E}}
\end{align}
where \(||A||\) is the norm of the linear map represented by \(A\).
An algorithm yielding an \((\alpha, \beta)\)-solution is called an \EMPH{\((\alpha, \beta)\)-solver}.

By arguments in~\cite{knp-pgtp-21}, we seek a preconditioner \(B\in \R^{V\times V}\) of full column rank such that, for any \(\Tilde{b}\in\R^V\) with \(\sum_{v\in V}{\Tilde{b_v}}=0\), it satisfies

\begin{equation}\label{eq:3.3}
||B\Tilde{b}||_1\le \Min\{||f||_{\dartsof{E}}:f\in \R^{\dartsof{E}}, Af=\Tilde{b}\}\le \kappa ||B\Tilde{b}||_1
\end{equation}
for some sufficiently small function \(\kappa\) of \(n\), \(\eps\), and \(d\).

Let \(M\) be the time it takes to multiply \(BA\) and \((BA)^T\) by a vector. Then there exists a \((1+\eps, \beta)\)-solver for any \(\eps, \beta>0\) for this problem with running time bounded by \(O(\kappa^2(|V|+|\dartsof{E}|+M)\log{|\dartsof{E}|}(\eps^{-2}+\log{\beta^{-1}})\)~\cite{s-gpumf-17}. Moreover, if a feasible flow \(f\in \R^{\dartsof{E}}\) with cost \(||f||_{\dartsof{E}}\le \kappa ||B\Tilde{b}||_1\) can be found in time \(K\), there is a \((\kappa, 0)\)-solver with running time \(K\).
By setting \(\beta=\eps\kappa^{-2}\)~\cite{knp-pgtp-21}, the composition of these two solvers is a \((1+2\eps, 0)\)-solver with running time bounded by

\begin{displaymath}
O(\kappa^2(|V|+|\dartsof{E}|+M)\log{|\dartsof{E}|}(\eps^{-2}+\log{\kappa})+K).
\end{displaymath}

\subsection{Legal shifts, blobs, and probabilities}
Our preconditioner~\(B\) essentially simulates the effects of randomly shifting~\(T\) along a diagonal.
For each level of~\(T\), we need to determine the probabilities that a random shift, conditioned on
cell boundaries not touching any moats, causes certain cells to contain certain subsets of vertices.
In order to do so efficiently, we build two sets of data structures.
One set helps us determine the set of legal shifts for~\(T\) relative to each dimension, and the other set
helps us maintain collections of points that cannot be separated into distinct cells during a legal
shift.
The second set of data structures will prove useful for efficiently determining how much flow to
send from several vertices at once as figuring these flows out for each individual vertex would be
too expensive.

Let \(\Delta_{\ell} = \Delta_{C^*}/(2^{\ell}e)\) be the lower bound on the side length of any
level~\(\ell\) cell as established in Lemma~\ref{lem:level_length}.
For each level \(\ell\), we compute a maximal set \(\S_{\ell_{i}}\) of values in \([0,
\Delta_{\ell})\) per dimension \(i\in \D\) such that no grid lines clip any moat of size
\(\frac{2\Delta_{\ell}}{n^2}\) around any vertex \(v\in V\) when we shift the grid by any value in
\(\S_{\ell_{i}}\).
We call \(\S_{\ell_{i}}\) the \EMPH{legal shifts} at level \(\ell\) relative to dimension \(i\), and we use
\(\S_{\ell} := \cap_i \S_{\ell_{i}}\) to denote the collection of legal shifts relative to all dimensions.

For computing \(\S_{\ell}\), we define a \EMPH{blob} at level \(\ell\) as a maximal set of points
that are guaranteed to be in the same cell at this level in our quadtree after an arbitrary shift.
Now we describe how to compute the set of all blobs \(\blb_{\ell}\) at each level \(\ell\).
Because the moat size is equal to \(\frac{2\Delta_{\ell}}{n^2}\), a blob can only become larger when \(\ell\) decreases.
Therefore, we compute each \(\blb_{\ell}\) in decreasing order of~\(\ell\) by combining smaller
blobs to bigger ones as we consider each level.
For any net point \(N_C\) and level \(\ell < \ell_C\), let \(\blb_{\ell}(N_C)\) denote the blob
containing \(N_C\) at level \(\ell\).
We leave \(\blb_{\ell}(N_C)\) undefined for any \(\ell \geq \ell_C\).
Let \(\bdx_{\blb}\) be a \emph{near}-minimum bounding box of all points in a blob \(\blb\).
More specifically, if blob \(\blb\) is at level \(\ell\), we make \(\bdx_{\blb}\) the box obtained
by extending the minimum bounding box of \(\blb\) by \(\frac{\Delta_{\ell}}{n^2}\) in each
dimension.
We use \(\bdx^{l, i}_{\blb}\) and \(\bdx^{r, i}_{\blb}\) to denote the coordinates of left and right
sides of \(\bdx_{\blb}\) for each dimension \(i\).

We now describe a way to compute \(\blb_{\ell}\) given \(\blb_{\ell+1}\).
First, for every \(\blb\in \blb_{\ell+1}\), we extend \(\bdx_{\blb}\) so it is now a near-minimum
bounding box of \(\blb\) at level \(\ell\).
Then, we sort blobs in \(\blb_{\ell+1}\) by the least coordinate of their bounding boxes in the
first dimension.
After that, we split these blobs to subsets by putting any two blobs \(\blb1\) and \(\blb2\) to the
same subset if \([\bdx^{l, i}_{\blb1}, \bdx^{r, i}_{\blb1}]\) and \([\bdx^{l, i}_{\blb2}, \bdx^{r,
i}_{\blb2}]\) are not disjoint, because we cannot put a grid line between these two blobs without
hitting the moat around at least one point in them.
For each of these subsets, we recursively perform this procedure for the remaining \(d-1\)
dimensions.
Every subset at the lowest level of recursion is a blob in \(\blb_{\ell}\).
If \(\ell\) is the largest level, we may assume \(\blb_{\ell+1}\) is the set of blobs where each
blob contains a single vertex in \(V\).
For each blob \(\blb \in \blb_{\ell}\), we call the blobs at level \(\ell+1\) inside it its \EMPH{child blobs}.
We use \(\cld(\blb)\) to denote the set of child blobs of \(\blb\) and \(\prnt(\blb)\) to denote the parent blob of \(\blb\).
Since the sets of points in blobs at the same level are disjoint, it is easy to see that the number
of distinct blobs among all levels together is at most \(2|V|-1\).
Now, consider a blob~\(\blb\) for some level~\(\ell\), and let~\(C\) be the level~\(\ell\) cell
containing~\(\blb\).
By the definition of~\(T\), there is an ancestor of~\(C\) at most~\(O(\log n)\) levels closer to the
root that has a sibling cell~\(C'\).
Cell~\(C'\) contains at least one blob~\(\blb'\).
Within another~\(O(\log n)\) levels, blobs \(\blb\) and \(\blb'\) will have near-bounding boxes
large enough to touch.
Therefore, every blob can only appear in at most \(O(\log{n})\) levels.
We define the \(\EMPH{blob forest}\) as the hierarchical structure of blobs defined above.  For
simplicity, we allow the same blob (with the same set of points) to appear multiple times in the
blob forest, once per level it appears.
The blob forest has \(O(|V|\log{n})\) nodes in total.

We now compute the legal shifts at each level \(\ell\) using blobs in \(\blb_{\ell}\).
To compute \(\S_{\ell_i}\) for some dimension \(i\), we look at \(\bdx^{l, i}_{\blb}\) and \(\bdx^{r, i}_{\blb}\) for \(\blb\in \blb_{\ell}\).
Let \(C\) be the cell that contains \(\blb\) at level \(\ell\).
Let \(\coor_{l,i}(C)\) and \(\coor_{r,i}(C)\) denote the coordinates of the left and right sides of \(C\) in dimension \(i\), respectively.
Then we call \([0,\Delta_{\ell}) \cap (\bdx^{l, i}_{\blb} - \coor_{l, i}(C), \bdx^{r, i}_{\blb} - \coor_{l, i}(C))\) the set of \EMPH{forbidden shifts} for \(\blb\).
Naturally, the set of legal shifts \(\S_{\ell_i}\) is equal to \([0, \Delta_{\ell})\) minus the union of forbidden shifts of all blobs in this dimension.
We can store \(\S_{\ell_i}\) in an array of size at most the number of blobs in \(\blb_{\ell}\) such that every element in the array is a maximal continuous subset in the union.
From now on, we assume elements in \(\S_{\ell_i}\) is sorted in increasing order by their lower bounds.
Therefore, we can precompute the total size of legal shifts before any element and then query the size of all legal shifts in \([x, y] \cap \S_{\ell_i} \) for any values \(x,y\) in \(O(\log{n})\) time.
The construction of all data structures mentioned above can be accomplished in~\(O(|V| \log^2 n)\)
time total through careful use of dynamic ordered dictionaries such as balanced binary search trees.

The last preparation for constructing the preconditioner is to compute the probability that a cell \(C\) contains a blob \(\blb\) if we shift the grid using a random value in \(\S_{\ell_C}\), for every pair of \(\blb\) and \(C\) at the same level.
We use \(\probability[\blb\in C]\) to denote this probability.
Recall, the side length of any cell at level \(\ell\) is at least \(\Delta_{\ell}\).
Let \(C_{\ell, \blb}\) be the cell containing \(\blb\) at level \(\ell\).
If we consider the legal shifts putting \(\blb\) in different cells in increasing order, we see each dimension is crossed at most once.
Therefore, there are at most \(d + 1\) cells for which we need to calculate the probability for each blob per level.
Let \(\cells_{\ell, \blb}\) denote this subset of cells.
Starting from the root level, for every level \(\ell\), we process the cells \(\cells_{\ell,
\blb}\).
Not every cell in~\(T\) has all~\(2^d\) possible children, so some of the left neighbors of \(C_{\ell, \blb}\) in the grid at level~\(\ell\) may not exist in~\(T\) itself.
For simplicity, we put a \EMPH{softlink} to \(C_{\ell, \blb}\) in place of such a grid cell if it
does not already have one.
For each \(C \in \cells_{\ell,\blb}\), we define \(\S_{C,\blb,i}\) to be the subset of \(\S_{\ell_i}\) that could make \(C\) cover \(\blb\) in dimension \(i\).
Let \(l_{\ell} := |\S_{\ell}|\) and \(l_{C, \blb} := |\cap_i \S_{C,\blb,i}|\).
We have \(\probability[\blb\in C] = \frac{l_{C, \blb}}{l_{\ell}}\).

\subsection{The preconditioner}

The preconditioner \(B\) is the \(V\times V\) matrix defined as follows.
For every net point \(u = N_C\), for every level \(\ell < \ell_C\), we let
\(\probability'[\blb_{\ell}(u) \in C]\) denote the sum of values \(\probability[\blb_{\ell}(u) \in
C']\) for all \(C'\) equal to or softlinked to~\(C\).
We set \(B_{N_{C''}, u} = \frac{\Delta_{\ell_{C}}}{3\Lambda}\cdot \probability'[\blb_{\ell}(u)\in
C'']\) where \(\Lambda=48d^{3/2}e^2\lg{n}\), for each cell \(C''\in \cells_{\ell,
\blb_{\ell}(u)}\) that is part of~\(T\).
In addition, we set \(B_{u, u} = \frac{\Delta_{\ell_{C}}}{3\Lambda}\) for all \(u \in V\) and set
all other entries to~\(0\).

Observe how for each column~\(N_C\) of~\(B\), the entries for each row~\(N_{C'}\) with \(\ell_{C'}
\geq \ell_C\) are all~\(0\), with the exception of row~\(N_C\) itself.
Any linear combination of columns excluding column~\(N_C\) with \(B_{N_C,N_C} =
\frac{\Delta_{\ell_C}}{3\Lambda}\) must have at least one non-zero value for some row~\(N_{C'},
\ell_{C'} > \ell_C\), implying the combination does not equal column~\(N_C\).
Matrix~\(B\) has full column rank.

We now describe an oblivious greedy algorithm that computes a flow~\(f\) such that~\(Af =
\Tilde{b}\) and the cost is at most \(\Lambda\) times that of the minimum cost.
This algorithm is used in the algorithm explicitly as the \((\kappa, 0)\)-solver discussed above,
and its existence is also used in establishing the condition number of~\(BA\).
We treat each blob as if it is moving the total divergence of higher level constituent vertices
together up toward the root.
By the time all the divergences reach the root, they will cancel each other out and the flow will be
valid for the vector~\(\Tilde{b}\).
For a blob \(\blb \in \blb_{\ell}\), define \(b_{\blb} := \sum_{N_C \in \blb \mid \ell_C \geq \ell}
\Tilde{b}_{N_C}\).
Observe, \(b_{\blb} = \sum_{\blb' \in \cld(\blb)} b_{\blb'} + \sum_{N_C \in \blb \mid \ell_C = \ell
} \Tilde{b}_{N_C}\).
To aid in moving divergences we treat the total divergence of each child blob of \(\blb\) as
a separate commodity.
The flow along each edge will be the sum of flows of all commodities on that edge.
Precisely, for every cell \(C\) in a postorder traversal of~\(T\), for every \(\blb\) with
\(\probability'[\blb\in C] > 0\), for every cell \(C'\in \cells_{\ell_C - 1, \prnt(\blb)}\), we add
\(\probability'[\prnt(\blb)\in C']\probability'[\blb \in C] b_{\blb}\) units of flow along the
unique path from~\(N_C\) to its parent to~\(N_C'\).
Observe that for any blob \(\blb \in \blb_{\ell}\) for some level~\(\ell\), we send \(b_{\blb}\)
units total to level \(\ell - 1\) cells.

We now establish both the approximation ratio of the greedy algorithm and the condition number
of~\(BA\).
Let \(f^*_{\Tilde{b}}:=\Argmin_{f\in\R^{\dartsof{E}}, Af=\Tilde{b}}{||f||_{\dartsof{E}}}\).
We arbitrarily decompose \(f^*_{\tilde{b}}\) into a set of flows \(F = \Set{f^1, f^2, \dots}\) with
the following properties: 1) each flow follows a simple path between two vertices \(u\) and \(v\);
2) for each flow \(f^i \in F\) and edge \((u, v) \in \dartsof{E}\) either \(f^i(u, v) = 0\) or its
sign is equal to the sign of \(f^*_{\tilde{b}}(u,v)\); 3) for each flow \(f^i \in F\) and vertex
\(v\), either \((Af^i)_v = 0\) or its sign is equal to that of \(\tilde{b}_v\); and 4) for each edge \((u,v) \in \dartsof{E}\), we have \(f^*_{\tilde{b}}(u, v) = \sum_{f^i \in F} f^i(u, v)\).
The existence of such a decomposition is a standard part of network flow theory.

Let~\(f\) be the flow found by our greedy algorithm.
We charge a portion of~\(||f||_{\dartsof{E}}\) to~\(||f^i||_{\dartsof{E}}\) for each flow~\(f^i\) so
that the sum of charges over all choices for~\(f^i\) sum to at least~\(||f||_{\dartsof{E}}\) and for
any one~\(f^i\), we overcharge by a factor of at most~\(\Lambda\).
Fix some \(f^i\in F\) sending flow from some vertex \(u\) to some vertex \(v\).
Let \(b^i_u = Af^i_u\).
We define \(b_u(C)\) as the part of divergence \(b^i_u\) that the greedy algorithm sends to \(C\).
Observe~\(b^i_u = -b^i_v\).
Without loss of generality we assume \(b^i_u\ge 0\) and \(b^i_v\le 0\).

Now we are ready to give the main lemma of the greedy algorithm.

\begin{lemma}\label{lem:surviving-b}
  Let \(C_u\) and \(C_v\) be the leaf cells containing \(u\) and \(v\), respectively.
  Let \(\ell\) be any level with \(\ell \le \min\{\ell_{C_u}, \ell_{C_v}\}\).
  Then the total amount of \(b^i_u\) not cancelled out by \(b^i_v\) at level \(\ell\) is
  \begin{align*}
      \sum_{C\in T, \ell_C = \ell} \max\{0, (b_u(C) + b_v(C))\} \le \frac{4d||u-v||_2}{\Delta_{\ell}} b^i_u.
  \end{align*}
\end{lemma}
\begin{proof}
  \def\leftcells{\mathbb{L}}
  \def\rightcells{\mathbb{R}}
  Let \(\leftcells_j \subseteq \cells_{\ell, \blb_{\ell}(u)}\) denote the subset of cells in \(\cells_{\ell, \blb_{\ell}(u)}\) that have lesser coordinates in dimension~\(j\), and let \(\rightcells_j := \cells_{\ell, \blb_{\ell}(u)} \setminus \leftcells_j\).
  If \(\cells_{\ell, \blb_{\ell}(v)} \cap \leftcells_j = \emptyset\) or \(\cells_{\ell, \blb_{\ell}(v)} \cap \rightcells_j = \emptyset\) for any \(j\), then  \(||u-v||_2\ge |\coor_j(u) - \coor_j(v)|\ge (1-2/n)\Delta_{\ell}\) and the lemma holds.
  From here on, we assume both \(\cells_{\ell, \blb_{\ell}(v)} \cap \leftcells_j\) and \(\cells_{\ell, \blb_{\ell}(v)} \cap \rightcells_j\) are non-empty for all dimensions \(j\).

  Let \(b_u(L_j) := \sum_{C \in \leftcells} b_u(C)\) be the total amount of \(b^i_u\) sent to \(\leftcells_j\), and define \(b_u(R_j)\), \(b_v(L_j)\), and \(b_v(R_j)\) similarly.
  There are~\(O(|V| \log n)\) nodes in the blob forest, and we may assume \(|V| \leq O(n \log n)\).
  Therefore, \(|s_{\ell}| \geq \Delta_{\ell} - \frac{d \cdot O(n \log^2 n) \Delta_{\ell}}{n^2} \geq \Delta_{\ell} / 2\), assuming \(n\) is sufficiently large.
  Let \(l_{\leftcells_j,\blb_{_\ell}(u)}\) denote the total length of legal shifts in \(s_{\ell}\) that make \emph{any} cell of \(\leftcells_j\) cover \(\blb_{\ell}(u)\)
  and define \(l_{\leftcells_j,\blb_{\ell}(v)}\) similarly.
  We have \(l_{\leftcells_j,\blb_{\ell}(u)} - l_{\leftcells_j,\blb_{\ell}(v)} \leq ||u-v||_2\).
  Therefore,
  \[
    b_u(L_j) + b_v(L_j) = \frac{l_{\leftcells_j,\blb_{\ell}(u)} - l_{\leftcells_j,\blb_{\ell}(v)}}{s_{\ell}} b^i_u 
    \leq \frac{2||u - v||_2}{\Delta_{\ell}} b^i_u.
  \]
  Similarly, \(b_u(R_j) + b_v(R_j) \leq 2||u - v||_2 / \Delta_{\ell} \cdot b^i_u\).

  Finally, summing over all dimensions \(j\), we have
  \begin{align*}
      \sum_{C\in T, \ell_C = \ell} \max\{0, (b_u(C) + b_v(C))\} &\leq \sum_{j \in \D} \max\{0, b_u(L_j) + b_v(L_j)\} + \max\{0, b_u(R_j) + b_v(R_j)\}\\
      &\leq \frac{4d||u - v||_2}{\Delta_{\ell}} b^i_u.
  \end{align*}

\end{proof}

\begin{lemma}\label{lem:total-overcharge}
  The flow computed by the greedy algorithm overcharges the cost of \(f^i\) by a factor of at most \(48d^{3/2}e^2\lg{n}\).
\end{lemma}
\begin{proof}
  Let \(\ell\) be the smallest value such that \(4d||u-v||_2\ge \Delta_{\ell}\).
  By Lemma~\ref{lem:surviving-b}, the divergence of \(b^i_u\) remaining at level \(\ell\) and
  greater is \(b^i_u\).
  The divergence at each level \(\ell'\) is sent to net points at level \(\ell'-1\) through paths of
  length at most \(3\sqrt{d}e^2\Delta_{\ell'}\).
  So the cost of carrying \(b^i_u\) to level \(\ell\) is at most
  \(\sum_{\ell'>\ell}3\sqrt{d}e^2\Delta_{\ell'} b^i_u \le 3\sqrt{d}e^2\Delta_{\ell} b^i_u\le 12d^{3/2}e^2||f^i||_{\dartsof{E}}\) in total.
  Starting from level \(\ell' = \ell\), the cost of carrying the remain divergence of \(b^i_u\) to
  one level less than the current level is at most \(\frac{4d||u - v||_2}{\Delta_{\ell'}} b^i_u
  \cdot 3\sqrt{d}e^2\Delta_{\ell'} \le 12d^{3/2}e^2||f^i||_{\dartsof{E}}\).
  Because \(\Delta_{\ell-1} > 4d||u-v||_2\), we have \(\Delta_{(\ell-1) - (2\lg{n}-2-\lg{d})} > n^2||u-v||_2\).
  Moats will force \(u\) and \(v\) to be in the same blob at any level \(\ell' \le (\ell-1) - (2\lg{n}-2-\lg{d})\).
  This means for any cell \(C\) with \(\ell_{C} \le (\ell-1) - (2\lg{n}-2-\lg{d})\), we have \(b_u(C) + b_v(C) = 0\).
  So we spend at most \(12d^{3/2}e^2||f^i||_{\dartsof{E}}\) cost to send divergence of \(b^i_u\) per level from level \(\ell\) to level \((\ell-1) - (2\lg{n}-2-\lg{d})\), and we spend zero cost for \(b^i_u\) after that.
  In total, the greedy algorithm charges at most \((2\lg{n}-\lg{d})12d^{3/2}e^2||f^i||_{\dartsof{E}}\) cost to divergence of \(u\) in \(f^i\).
  Similarly, we can show the cost it charges to divergence of \(v\) has the same upper bound, and the lemma holds.
\end{proof}

In our algorithm, \(\frac{|B\Tilde{b}_{N_C}|3\Lambda}{\Delta_{\ell_C}}\) divergence leaves the net point
\(N_C\) of a cell \(C\) through paths of length at most \(3\sqrt{d}e^2\Delta_{\ell_C}\).
  On the other hand, these divergences leave \(N_C\) through edges of length at least
  \(\sqrt{d}(1/2-(1/(n^2+1)))\Delta_{\ell_C}\).
All together, we see \(||B\Tilde{b}||_1 < \frac{||f||_{\dartsof{E}}}{\Lambda} \le
||f^*||_{\dartsof{E}}\).
Therefore, by setting \(\kappa = 9\sqrt{d}e^2\Lambda\), we have
\begin{equation*}
  ||B\Tilde{b}||_1 < \min\{||f||_{\dartof{E}}\,:\, f\in \R^{\dartof{E}}, Af=\Tilde{b}\} \le ||f||_{\dartsof{E}} \le\kappa||B\Tilde{b}||_1.
\end{equation*}

\begin{lemma}
  \label{lem:DP}
  Applications of \(BA\) and \((BA)^T\) to arbitrary vectors \(f\in \R^{\dartsof{E}}\) and \(\Tilde{b}\in \R^V\), respectively, can be done in \(O(|E|\log{n})\) time.
\end{lemma}
\begin{proof}
  Let \(A' = Af\) and let \(b' = B^T\Tilde{b}\).
  Both \(A, f\) and \(b'\) has \(O(|E|)\) non-zero entries, so we can compute \(A'\) and \(A^Tb'\) in \(O(|E|)\) time given \(b'\).
  We show how to compute \(BA'\) and \(B^T\Tilde{b}\) in \(O(|E|\log{n})\) time.
  \paragraph*{Computing \(BA'\)}
  Let \(C_u\) be the cell with \(u\) as its net point.
  By the definition of \(B\), for each cell \(C\),
  \begin{align*}
    (BA')_{N_C} &= \frac{\Delta_{C}}{3\Lambda}A'_{N_C} + \frac{\Delta_{C}}{3\Lambda}\sum_{u\in V, \probability'[\blb_{\ell_C}(u)\in C] > 0}\probability'[\blb_{\ell_C}(u)\in C]A'_u\\
                &= \frac{\Delta_{C}}{3\Lambda}A'_{N_C} +
                \frac{\Delta_{C}}{3\Lambda}\sum_{\blb\in\blb_{\ell_C}, \probability'[\blb\in C] > 0}{\probability'[\blb\in C]\sum_{u\in \blb, \ell_{C_u} > \ell}}A'_u
  \end{align*}
  There are at most \(2|V|-1\) different blobs and \(\sum_{u\in \blb, \ell_{C_u} > \ell} A'_u = \sum_{\blb'\in \cld(\blb)}\sum_{u\in \blb, \ell_{C_u}> \ell+1}A'_u + \sum_{u\in \blb', \ell_{C_u} = \ell+1, \blb'\in \cld(\blb)}{A'_u}\) for each blob at some level \(\ell\).
  So we can compute \(\sum_{u\in \blb, \ell_{C_u} > \ell}A'_u\) for each blob \(\blb\) in \(O(|V|\log{n})\) time in total during a postorder traversal of the blobs.
  After that, we can fill in all entries in \(BA'\) in \(O(|E|)\) time.

  \paragraph*{Computing \(b'\)}
  For every point \(u\), except \(B^T_{u,u}\), every non-zero entry in \(B^T_{u}\) corresponds to a net point \(N_C\) of a cell \(C\) with \(\probability'[\blb_{\ell_C}(u)\in C] > 0\).
  Let \(C_u\) be the cell with \(u\) as its net point.
  We have \(b'_u = \frac{\Delta_{C_u}}{3\Lambda}\Tilde{b}_u+\sum_{C, \probability'[\blb_{\ell_C}(u)\in C]> 0} \frac{\Delta_{C}}{3\Lambda}\probability'[\blb_{\ell_C}(u)\in C] \Tilde{b}_{N_C}\).
  Let \(\blb^-\) denote the set of strict ancestor blobs of a blob \(\blb\).
  Let \(\ell\) be any level where \(\blb_{\ell}(u)\) is defined.
  We have \(\sum_{C, \ell_C \le \ell, \probability'[\blb_{\ell_C}(u)\in C]> 0}
  \frac{\Delta_{C}}{3\Lambda}\probability'[\blb_{\ell_C}(u)\in C] \Tilde{b}_{N_C} = \sum_{C, \ell_C
  = \ell, \probability'[\blb_{\ell}(u)\in C]>
  0}\frac{\Delta_{C}}{3\Lambda}\probability'[\blb_{\ell}(u)\in C] \Tilde{b}_{N_C}\linebreak
  + \sum_{\blb'\in \blb_{\ell}(u)^-} \sum_{C, \probability'[\blb' \in C]> 0}\frac{\Delta_{C}}{3\Lambda}\probability'[\blb'\in C] \Tilde{b}_{N_C}\).
  We can compute \(\sum_{\blb'\in \blb^-} \sum_{C, \probability'[\blb' \in C]> 0}\frac{\Delta_{C}}{3\Lambda}\probability'[\blb'\in C]\Tilde{b}_{N_C}\) for each blob \(\blb\) in \(O(|V|\log{n})\) time in total during a preorder traversal of the blobs.
  Then we can fill in each entry of \(b'\) in constant time.
\end{proof}

We have shown there exists a \((1+2\eps, 0)\)-solver for the minimum cost flow problem on \(G\).
Plugging in all the pieces, we see the running time of the solver is at most
\(O(|E|\eps^{-2}\log^4{n}\log{\log{n}})\).

\section{Recovering a transportation map from a flow}
\label{sec:recover}
Let~\(\hat{G} = (V, E)\) be \emph{any} connected graph such that~\(P \subseteq V \subset \R^d\) and
each edge has weight equal to the Euclidean distance between endpoints.
Let~\(A\) be the vertex-edge incidence matrix of~\(\hat{G}\), and let~\(\hat{f} \in
\R^{\dartsof{E}}\) be any flow in~\(\hat{G}\) such that \(A\hat{f} = \supply\) where \(\supply(v) =
0\) if \(v \notin P\).
In this section, we show how to transform~\(\hat{f}\) into a transformation map for \((P, \supply)\)
where \(\cost(\map) \leq ||\hat{f}||_{\dartsof{E}}\).
Throughout this section, we let~\(m = |E|\).
We also assume~\(m = O(n^4)\), as we could simply compute an \emph{optimal} transportation map from
scratch otherwise using an algorithm for minimum cost flow in general graphs~\cite{o-fspmc-93}.

Let~\(\dartsof{E}'\) denote the edges of the complete graph over~\(V\) where each edge is oriented
consistently with its counterpart in~\(\dartsof{E}\) if it exists and oriented arbitrarily
otherwise.
We maintain a flow~\(f \in \R^{\dartsof{E}'}\) where initially \(f_{(u,v)} = \hat{f}_{(u,v)}\) if
\(uv \in E\), and \(f_{(u,v)} = 0\) otherwise.
We will eventually guarantee \(f_{(u,v)} \neq 0\) only for \(u,v \in P\).

For each point~\(p \in P\), there are potentially~\(\Theta(|E|)\) other vertices that may at some
point during the process directly send flow to or receive flow from~\(p\).
We cannot afford to update these flow assignments individually, so for each vertex \(v \in V\),
we instead maintain a single \EMPH{prefix split tree}~\cite{fl-ntasg-22} \(\splittree(v)\) that
will contain representations of certain vertices sending flow to~\(v\).
A prefix split tree~\(\splittree\) is an ordered binary tree where each node~\(\splitnode\) is
assigned a non-negative \EMPH{potential}~\(\potential(\splitnode)\).
We let~\(\potential(\splittree)\) denote the total potential of nodes in~\(\splittree\).
Prefix split trees containing~\(s\) nodes support the following operations in amortized \(O(\log
s)\) time each:
\begin{itemize}
  \item
    \(\textsc{Insert}(\splittree, \potential)\): Insert a node of potential~\(\potential\) into
    tree~\(\splittree\) and return a reference to this node.
  \item
    \(\textsc{Delete}(\splittree, \splitnode)\): Delete the node~\(\splitnode\) from the
    tree~\(\splittree\).
  \item
    \(\textsc{Merge}(\splittree, \splittree')\): Modify \(\splittree\) by adding all nodes of
    \(\splittree'\) after the nodes of \(S\), emptying \(\splittree'\) in the process.
  \item
    \(\textsc{PrefixSplit}(\splittree, t)\): Assume~\(0 \leq t \leq \potential(\splittree)\).
    If a prefix of nodes in \(\splittree\) has total potential exactly~\(t\), then
    let~\(\splitnode_1\) be the last member of this prefix.
    Otherwise, let~\(\splitnode\) be the first node where the prefix through~\(\splitnode\) has
    total potential greater than~\(t\).
    \EMPH{Split} \(\splitnode\) by replacing it in-place with two nodes~\(\splitnode_1\)
    and~\(\splitnode_2\) such \(\potential(\splitnode_1) + \potential(\splitnode_2) =
    \potential(\splitnode)\) and the prefix through~\(\splitnode_1\) has total potential
    exactly~\(t\).
    Either way, create a new tree~\(\splittree'\) containing the prefix through~\(\splitnode_1\) and
    remove this prefix from~\(\splittree\).
\end{itemize}

Each node~\(\splitnode\) in \(\splittree(v)\) represents a vertex~\(u \in V\), and each vertex may
be represented by multiple nodes, even within a single prefix split tree.
We denote the vertex represented by~\(\splitnode\) as~\(\splitpoint(\splitnode)\).
All split tree~\(\splittree(v)\) are initially empty.
When a node \(\splitnode\) is split into two nodes \(\splitnode_1\) and~\(\splitnode_2\), we set
\(\splitpoint(\splitnode_1) = \splitpoint(\splitnode_2) := \splitpoint(\splitnode)\).
Along with the prefix split trees, we maintain a so-called \EMPH{base flow} \(f' \in
\R^{\dartsof{E}}\) that is initially equal to~\(f\).
We maintain the invariant that for each pair of vertices~\(u\) and~\(v\), \(f'_{(u,v)} +
\sum_{\splitnode \in \splittree(v) \mid \splitpoint(\splitnode) = u} \potential(\splitnode) =
f_{(u,v)}\).

The \EMPH{support} of flow~\(f'\) is the set of undirected edges \(uv\) for which \(f'_{(u,v)} \neq
0\).
We begin by changing~\(f'\) and therefore~\(f\) so that its support is a forest.
We use a process inspired by the acyclic flow algorithm of Sleator and Tarjan~\cite{st-dsdt-83}.
Let~\(\bar{G} = (V, \bar{E})\) be initially empty.
We iteratively process each directed edge~\((u,v) \in \dartsof{E}\) such that~\(f'_{(u,v)} > 0\).
When it comes time to process~\((u,v)\), we check if~\(u\) and~\(v\) are in the same component
of~\(\bar{G}\).
If not, we add~\(uv\) to~\(\bar{G}\).
Otherwise, let~\(\pi\) be the directed path from~\(u\) to~\(v\) in~\(\bar{G}\).
We define the \EMPH{unit cost} of~\(\pi\) to be~\(|\pi| := \sum_{(x,y) \in \pi \mid f_{(x,y)} > 0}
||y - x||_2 - \sum_{(x,y) \in \pi \mid f_{(x,y)} < 0} ||y - x||_2\), the amount
\(||f||_{\dartsof{E}}\) increases per unit of flow added to the directed edges of \(\pi\).
If \(||v - u||_2 \geq |\pi|\), let \((o,p) = \argmin_{(x,y) \in \Paren{\Set{(x',y') \in \pi \mid
f_{(x',y')} < 0} \cup (v,u)}} -f_{(x,y)}\), the first edge to go to~\(0\) flow if we ``reroute`` as
much flow along~\(\pi\) instead of~\((u,v)\) as we can, and let~\(F = f_{(o,p)}\).
If \(||v - u||_2 < |\pi|\), let \((o,p) = \argmin_{(x,y) \in \pi \mid f_{(x,y)} > 0} f_{(x,y)}\) and
\(F = -f_{(o,p)}\) instead.
We modify~\(f'\) by increasing the flow along all directed edges of~\(\pi\) by~\(F\) and decreasing
the flow along~\((u,v)\) by~\(F\).
Doing so causes~\(f'_{(o,p)} = 0\).
If \(op \neq uv\), we remove~\(op\) from~\(\bar{G}\) and add~\(uv\) in its place.
We are now done processing~\(uv\).
Observe~\(\bar{G}\) remains a forest after processing each edge.
Therefore, each edge can be processed in (amortized) \(O(\log n)\) time using standard extensions to
dynamic tree data structures~\cite{st-dsdt-83}.

\begin{lemma}\label{lem:support_forest}
  The above procedure does not change~\(Af\), the cost of~\(||f||_{\dartsof{E}}\) does not increase,
  and the support of~\(f\) becomes a forest.
  Further, if~\(\supply(p)\) is an integer for all~\(p \in P\), then the procedure
  guarantees~\(f_{(u,v)}\) is an integer for all~\(uv \in E\).
\end{lemma}
\begin{proof}
  Each time the flow~\(f'\), and thus \(f\), are changed, we do so by changing the route some flow
  takes between a pair of vertices~\(u\) and~\(v\).
  We change the flow along the path~\(\pi\) by the opposite amount we change~\(f'_{(u,v)}\), so the
  vector~\(Af\) does not change.
  Further, the choice to increase or decrease flow along~\(\pi\) is made so that the change cannot
  increase~\(||f||_{\dartsof{E}}\).
  Whenever an edge~\(uv\) is about to be added to~\(\bar{G}\) and create a cycle, we remove an edge
  (possibly \(uv\) itself) from that cycle.
  Therefore~\(\bar{G}\) and the support of~\(f\) is a forest.

  For the claim about \(f\) being integral, observe that it is trivially true if every component
  of~\(\bar{G}\) contains one vertex.
  If some component contains multiple vertices, let~\(u\) be a leaf in that component, and
  let~\(uv\) be its one incident edge.
  Because~\((Af)_u\) is integral, \(f_{(u,v)}\) must be integral as well.
  If we (for the sake of proof) remove~\(u\) from~\(\bar{G}\) and set~\(f_{(u,v)} = 0\), then
  \((Af)_v\) remains integral.
  The claim follows by induction on the number of vertices in~\(\bar{G}\).
\end{proof}

Consider the orientation of~\(\bar{G}\) such that for each directed edge \((u,v)\) in the
orientation, \(f_{(u,v)} > 0\).
We now process each vertex in topological order with respect to this orientation.

Suppose it is time to start processing vertex~\(v\).
Our procedure guarantees that 1) \(f'_{(u, v)} = 0\) for each vertex~\(u\) that has already been
processed, 2) \(f'_{(v, w)}\) has not yet changed for each vertex~\(w\) that we have not yet
processed, and 3) \(v\) is not yet represented in \(\splittree(w)\) for \emph{any} vertex~\(w\).

From the above guarantees and the definition of~\(\bar{G}\), we may conclude that~\(f'_{(v, w)} \geq
0\) for any vertex~\(w\) we have not yet processed.
Our goal is to shortcut flow passing through \(v\) from a processed vertex \(u\) to an unprocessed
vertex \(w\) by adding to \(w\)'s split tree.
If \(v \in P^+\), then let \(\splitnode\) be the node returned by
\(\textsc{Insert}(\splittree(v), \supply(v))\) and set \(\splitpoint(\splitnode) \gets v\).
In doing so, we're implicitly declaring that~\(v\) is receiving~\(\supply(v)\) units of flow from
itself, and we don't have to set up any special cases for when we want \(v\) to actually send flow.
This moment is the only time we create new nodes for the split trees.
Observe whether or not we create a new node, we now have \(\potential(\splittree(v)) \geq \sum_{w
\in V} f'_{(v,w)}\).

While there exists a vertex~\(w\) such that \(f'_{(v,w)} > 0\), we do the following.
Let~\(\splittree'\) be the tree returned by \(\textsc{PrefixSplit}(\splittree(v), f'_{(v,w)})\).
We perform a \(\textsc{Merge}(\splittree(w), \splittree')\), shortcutting the flow through \(v\) to
\(w\) as desired.
Finally, we set~\(f'_{(v,w)} \gets 0\) as all flow into~\(w\) originally from~\(v\) is now
represented in~\(\splittree(w)\).
We are done processing~\(v\) when the while loop concludes.
We may easily verify each of our guarantees hold for later vertices in the topological order.

Consider when we have finished processing all the vertices.
Those vertices~\(v \in P^+\) are represented as one or more nodes in the vertices' split trees, and
these nodes have total potential \(\supply(v)\).
Those vertices~\(v \in P^-\) each have a split tree of total potential \(\potential(\splittree(v)) =
-\supply(v)\).
We now construct the transportation map~\(\map\).
Initially~\(\map(u, v) = 0\) for all \((u,v) \in P^+ \times P^-\).
While there exists a split tree~\(\splittree(v)\) containing at least one node~\(\splitnode\), we
increase~\(\map(\splitpoint(\splitnode), v)\) by \(\potential(\splitnode)\) and perform a
\(\textsc{Delete}(\splittree(v), \splitnode)\).
When the loop completes, we are done constructing~\(\map\).

\begin{lemma}\label{lem:recover}
  The algorithm above produces a transportation map~\(\map\) for \((P, \supply)\) such that
  \(\cost(\map) \leq ||\hat{f}||_{\dartsof{E}}\) in \(O(m \log n)\) time.
  Further, if~\(\supply(p)\) is an integer for all~\(p \in P\), then the procedure
  guarantees~\(\map(p,q)\) is an integer for all~\((p,q) \in (P^+ \times P^-)\).
\end{lemma}
\begin{proof}
  The fact that~\(\map\) is a transportation map for~\((P, \supply)\) follows from the above
  discussions.
  Observe that every time we change~\(f\) while processing vertices in topological order, we do so
  by rerouting flow going from some vertex~\(u\) through a vertex~\(v\) and then to a vertex~\(w\).
  By triangle inequality, these shortcuts can only reduce the cost of~\(f\), implying our bound on
  \(\cost(\map)\).
  If \(\supply(p)\) is integral for all \(p \in P\), then \(f_{(u,v)}\) is integral immediately
  before we start processing vertices in topological order.
  Each change reroutes an amount of flow equal to the flow along an edge, so the flow values remain
  integral.

  For running time, we observe we perform a constant number of split tree operations for each of
  the~\(m\) or fewer edges in~\(\bar{G}\) while processing the vertices in topological order.
  These operations takes~\(O(m \log n)\) time total.
  We then do a number of split trees operations equal to the total number of nodes in all split
  trees while adding values to pairs in the transportation map~\(\map\).
  The only operations that can add nodes to a split tree are the \(\textsc{Insert}\)s done for each
  vertex of positive supply, and the \(\textsc{PrefixSplit}\)s done for each edge in~\(\bar{G}\).
  Therefore, we create~\(O(m)\) nodes total and remove them from the split trees in~\(O(m \log n)\)
  time.
  Adding in the~\(O(m \log n)\) time needed to construct and topologically sort~\(\bar{G}\), we
  conclude our proof of the running time.
\end{proof}

We are now ready to state and prove our main theorem.
\begin{theorem}
  There exists a deterministic algorithm that, given a set of~\(n\) points \(P \subset \R^d\) and a
  supply function \(\supply : P \to \R\), runs in time \(O(n \eps^{-(d+2)} \log^5 n \log \log n)\)
  and returns a transportation map~\(\map\) with cost at most \((1 + \eps) \cdot \cost^*(P,
  \supply)\).
  Further, if~\(\supply(p)\) is an integer for all~\(p \in P\), then~\(\map(p,q)\) is an integer for
  all~\((p,q) \in (P^+ \times P^-)\).
\end{theorem}
\begin{proof}
  Recall, we build a warped quadtree~\(T\) while contracting certain subsets of~\(P\).
  Let~\((P', \supply')\) denote the geometric transportation instance after contraction and let~\(n'
  = |P'|\).
  We build the sparse spanner graph~\(G = (V, E)\) over~\(P'\) in~\(O(n' \eps^{-d} \log n)\) time.
  Let~\(m = |E| = O(n' \eps^{-d} \log n)\).
  We define an instance of uncapacitated maximum flow~\((G, b^*)\) where~\(b^*_v\) is equal to
  \(\supply'(v)\) if \(v \in P'\) and equal to \(0\) otherwise.
  By Lemma~\ref{lem:approximate_distances}, \(\cost^*(G, b^*) \leq (1 + O(\eps)) \cost^*(P',
  \supply')\).
  We compute a flow~\(f\) of cost \((1 + O(\eps)) \cdot \cost^*(G, b^*) = O(1 + O(\eps)) \cost^*(P',
  \supply')\) in \(O(m \eps^{-2} \log^4 n \log \log n) = O(n' \eps^{-(d+2)} \log^5 \log \log
  n)\) time using the algorithm described in Section~\ref{sec:flow}.

  By the discussion at the end of Section~\ref{sec:spanner}, we can combine the spanner~\(G\) and
  the flow~\(f\) with the recursively computed spanners' \((1 + O(\eps))\)-approximate flows for
  each contracted subset of~\(P\) to yield a flow~\(\hat{f}\) for a single spanner \(\hat{G}\) on
  \((P, \supply)\).
  This flow has cost \(||\hat{f}||_{\dartsof{E}} = (1 + O(\eps)) \cdot \cost^*(P, \supply)\).
  Finally, we compute a transportation map~\(\map\) for \((P, \supply)\) with cost at
  most~\(||\hat{f}||_{\dartsof{E}} = (1 + O(\eps)) \cost^*(P, \supply)\).
  If~\(\supply(p)\) is an integer for all~\(p \in P\), then~\(\map(p,q)\) is an integer for
  all~\((p,q) \in (P^+ \times P^-)\) per the above discussions.

  By Lemmas~\ref{lem:num_cells} and~\ref{lem:tree_time}, we spend \(O(n \log^2 n)\) time total
  constructing all warped quadtrees across the various recursive subproblems.
  We then spend \(O(n \eps^{-(d+2)} \log^5 n \log \log n)\) time computing flows for all individual
  subproblems and~\(O(n \log^2 n)\) time transforming the flows into a transportation map.
  We conclude our proof.
\end{proof}

\section{Simplifying the algorithm for low spread cases}
\label{sec:low_spread}
In this section, we sketch some simplifications that can be made to our algorithm for the case
that~\(\spread(P)\) is small.
Our simplified algorithm computes a \((1+\eps)\)-approximation of the optimal transportation map in
\(O(n\eps^{-(d+2)}(\log{n}+\log^3 \spread(P) \log{\log{\spread(P)}})\log{\spread(P)})\) time.
When \(\spread(P) = n^{O(1)}\), the running time of the simplified algorithm is slightly better
than the one for the unbounded spread case.

Instead of building a warped quadtree as in the first half of Section~\ref{sec:spanner}, we use a
standard quadtree~\(T\) where all cells at a level have exactly the same size and the leaves are
exactly those cells containing one point of~\(P\).
Therefore, we do not need the moat avoidance data structures.
There is no need to contract subsets of~\(P\), and the depth of \(T\) is \(\log{\spread{(P)}}+1\).
We build our sparse graph \(G = (V, E)\) on \(T\) using the procedure described in
Section~\ref{subsec:build-spanner}.
Lemma~\ref{lem:approximate_distances} still holds on \(G\).
The time to construct~\(T\) and~\(G\) is \(O(n \eps^{-d} \log \spread)\) and \(|E| = O(n \eps^{-d}
\log \spread)\) as well.

When finding the \((1 + O(\eps))\)-approximation for the minimum cost flow instance \((G, b^*)\), we
no longer worry about moats.
For the greedy algorithm and preconditioner in Section~\ref{sec:flow}, we essentially treat each
point \(u \in V\) as its own blob appearing at every level of the quadtree.
At level~\(\ell\), we allow \emph{all} shifts in \([0, \Delta^* / 2^\ell]^d\), thus eliminating the
need for the legal shift and blob data structures.
Lemma~\ref{lem:surviving-b} and Lemma~\ref{lem:total-overcharge} together imply the conditioner
number \(\kappa\) of the preconditioner in the low spread case is at most
\(144d^2\log{\spread(P)}\).
Therefore, we can compute a flow with cost at most \((1+O(\eps)) \cdot \cost^*(P, \supply)\) in
\(O(n\eps^{-(d+2)}(\log^4 \spread(P) \log \log \spread(P)))\) time using Sherman's preconditioner
framework.

Our procedure for recovering a transportation map from the flow is unchanged, running in
\(O(n\eps^{-d}\log n \log{\spread(P)})\) time.
Considering everything above, we get the following theorem.
\begin{theorem}
  There exists a deterministic algorithm that, given a set of~\(n\) points \(P \subset \R^d\) of
  spread \(\spread(P)\) and a supply function \(\supply : P \to \R\), runs in time\linebreak
  \(O(n\eps^{-(d+2)}(\log{n}+\log^3 \spread(P) \log{\log{\spread(P)}})\log{\spread(P)})\) and
  returns a transportation map~\(\map\) with cost at most \((1 + \eps) \cdot \cost^*(P, \supply)\).
  Further, if~\(\supply(p)\) is an integer for all~\(p \in P\), then~\(\map(p,q)\) is an integer for
  all~\((p,q) \in (P^+ \times P^-)\).
\end{theorem}

Recall, the geometric bipartite matching problem is the special case where~\(\supply(p) \in
\Set{-1,1}\) for all~\(P\), and the transportation map is required to assign either~\(0\) or~\(1\)
to each pair of points.
Approximating an arbitrary case of the geometric bipartite matching problem can be reduced in \(O(n
\log^2 n)\) time to an instance where the spread is polynomial in \(n\)~\cite{acrx-dnaag-22}.
As our algorithm is guaranteed to return a \(0,1\) map given such an instance, we conclude with the
follow corollary.

\begin{corollary}
  There exists a deterministic algorithm that, given an \(n\)-point instance of the geometric
  bipartite matching problem, returns a \((1 + \eps)\)-approximately optimal matching in time\linebreak
  \(O(n \eps^{-(d+2)} \log^4 n \log \log n)\).
\end{corollary}

\section{Uncapacitated minimum-cost flow in general graphs}
\label{sec:general_flow}

Previously, we reduced approximating the geometric transportation problem to approximating a special case of minimum-cost flow without edge capacities.
In this section, we turn the situation around by showing how to approximate minimum-cost flow in a \emph{general} graph via reductions to our algorithm for geometric transportation.
Our algorithm is based on the one given in~\cite{asz-pausp-20} for the case of moderate integer edge costs.

Let \(G = (V, E)\) be an arbitrary undirected graph, let \(|| \cdot ||_{\dartsof{E}}\) denote an arbitrary norm on \(\R^{\dartsof{E}}\), and let \(b \in \R^V\) denote an arbitrary set of vertex divergences.
In this section, we let \(n := |V|\) and \(m := |E|\).
Fix a parameter~\(\eps > 0\).
We again use Sherman's~\cite{s-gpumf-17} framework as described in Section~\ref{sec:flow}.
Accordingly, we need a preconditioner \(Q \in \R^{r \times V}\) of full column rank such that 
\begin{equation}\label{eq:kappa-general}
||Q\Tilde{b}||_1\le \Min\{||f||_{\dartsof{E}}:f\in \R^{\dartsof{E}}, Af=\Tilde{b}\}\le \kappa ||Q\Tilde{b}||_1
\end{equation}
for any \(\Tilde{b} \in \R^V\) and with \(\kappa\) small.
Note that \(r \neq |V|\) in this case;
we'll leave it unspecified for now.
We also need to describe an efficient \(\kappa\)-approximate ``oblivious'' greedy algorithm to help us estimate~\(\kappa\).
However, as in~\cite{asz-pausp-20}, we'll actually run iterations of Sherman's framework until it suffices to use a simple \(n\)-approximation to satisfy the final set of divergences.

We'll begin with the greedy algorithm as it makes it easier to describe the preconditioner itself.
We start with the following lemma.
\begin{lemma}[\cite{b-lefms-85}]\label{lem:Bourgain}
    There is a randomized algorithm which can output a mapping \(\psi: V \to \R^d\) with \(d = O(\log^2 n)\) with constant probability in \(O(m \log^2 n)\) time such that for all \(u, v \in V\)
    \[ \dist_G(u, v) \leq ||\psi(u) - \psi(v)||_2 \leq O(\log n) \cdot \dist_G(u, v). \]
\end{lemma}

A solution to the geoemtric transportation problem for \(\psi(V)\) should form a reasonable estimate of the cost of the optimal flow.
Unfortunately, the dimension of the target space is moderately large.
We can deal with the large dimensionality by using the following weakening of our main result.
\begin{theorem}
    \label{thm:main-simplified}
    Suppose \(d\) \emph{is not} constant.
    There exists a deterministic algorithm that, given a set of~\(n\) points \(P \subset \R^d\) and a
  supply function \(\supply : P \to \R\), runs in time \(O(d n \log n)\)
  and returns a transportation map~\(\map\) with cost at most \(O(d^2 \log n) \cdot \cost^*(P,
  \supply)\).
\end{theorem}
\begin{proof}
    We build the spanners as before, except we place leaves immediately when a cell contains exactly one point and add edges only between net points and the net points of their neighboring cells.
    The resulting spanners have \(O(d n \log n)\) vertices and edges total, and they maintain shortest path distances up to an \(O(\sqrt{d})\) factor.
    See Lemma~\ref{lem:approximate_distances}.

    We define the preconditioner \(B\) as before.
    A single iteration of the greedy algorithm results in an~\(O(d^{3/2} \log n)\) approximation to
    the spanner's minimum cost flow instance.
    See Lemma~\ref{lem:total-overcharge}.
    We run a single iteration of the greedy algorithm in each spanner in \(O(d n \log^2 n)\) time total, resulting in \(O(d^2 \log n)\) approximately optimal flows.
    We combine and transform them into proper transportation maps in \(O(d n \log^2 n)\) additional time as described in Section~\ref{sec:recover}.
\end{proof}

Our greedy algorithm for seeking approximately optimal flows on~\(G\) computes a Bourgain embedding as described in Lemma~\ref{lem:Bourgain}.
We can then use the algorithm of Theorem~\ref{thm:main-simplified} to get an~\(O(d^2 \log n) \cdot O(\log n) = \log^{O(1)} n\) approximation on the minimum-cost flow value for the original problem in \(n \log^{O(1)} n\) time.
Note that our algorithm for minimum cost flow need not actually extract a transportation map from the spanner flows.

We are now ready to describe the preconditioner~\(Q\) needed for the minimum-cost flow instance on~\(G\).
Let~\(V'\) denote the full set of net points in each of the spanners built by the algorithm of Theorem~\ref{thm:main-simplified}.
Let \(Q^1\) denote the \(|V'| \times |V|\) \(0-1\) matrix where \(Q^1_{\psi(u),N_{C_\psi(u)}} = 1\) for all vertices \(u \in V\), and all other entries are~\(0\).
Let \(Q^2\) denote the \(|V'| \times |V'|\) real-valued matrix composed of the spanners' individual preconditioner matrices where \(Q^2_{u, v} = B_{u, v}\) if \(u\) and \(v\) belong to the same spanner with preconditioner \(B\).
All other entries of \(Q^2\) are~\(0\).
Finally, let \(Q = Q^2 Q^1 \in \R^{V' \times V}\).

Matrix \(Q\) has full column rank.
The value of~\(\kappa\) in~\eqref{eq:kappa-general} is~\(\log^{O(1)} n\).
For any \(f \in \R^{\dartsof{E}}\), we can compute \(QAf\) in \(O(m\log n) + n \log^{O(1)} n\) time by first computing \(q' := Q^1 A f\) and then applying the algorithm of Lemma~\ref{lem:DP} to compute \(Q^2 q'\).
In fact, we can compute \(Q^2 q'\) in time proportional to the size of the spanners, because we no longer need to track which blob flow originally came from.
For any \(\Tilde{b} \in \R^V\), we can compute \((QA)^T \Tilde{b}\) in the same time by first computing \(b' := {Q^2}^T \Tilde{b}\) and then computing \(A^T {Q^1}^T b'\).

Assuming we compute a good embedding with Lemma~\ref{lem:Bourgain},
there exists a \((1 + \eps, \eps^{1 + \lg n} / \kappa)\)-solver for the minimum cost flow instance that performs \(\eps^{-2} \log^{O(1)} n\) matrix multiplications.
We can compose this solver with a simple \((n, 0)\)-solver that runs in \(m \log^{O(1)} n\) time to get a \((1 + O(\eps))\)-approximate solution to the minimum cost flow instance.
The total time spent is \(m \log^{O(1)} n\).
We can run our algorithm \(O(\log n)\) times to guarantee success with high probability \(1 - 1/n^c\) for any constant \(c\).

\begin{theorem}
There exists a randomized algorithm that, given an undirected graph \(G = (V, E)\) with \(n\) vertices and \(m\) edges, an arbitrarily norm on \(\R^{\dartsof{E}}\), and an arbitrarily set of vertex divergences \(b \in \R^V\), runs in time \(m \eps^{-2} \log^{O(1)} n\)
and returns a \((1 + \eps)\)-approximate uncapacitated minimum cost flow in~\(G\) with high probability.
\end{theorem}


\section*{Acknowledgments}

The authors would like to thank the anonymous reviewers of earlier versions of this paper for their helpful and insightful comments.

\bibliographystyle{alpha} 
\bibliography{refs}

\end{document}